%% file: draftGANDALF2010.tex
\documentclass{eptcs}
\input{generic}

\input{fixpseudocode}

\usepackage[OT1]{eulervm}
\usepackage{color}
\usepackage{times}
\usepackage{graphicx}
\usepackage{latexsym}
\usepackage{amssymb}
\usepackage{epic}
\usepackage{eepic}
\usepackage{fancybox}
\usepackage[english]{babel}
\usepackage{amsfonts}
\usepackage{amssymb}
\usepackage{amsmath}
\usepackage{enumerate}
\usepackage{tikz}
\usepackage{paralist}
\usepgflibrary{decorations.pathreplacing} % LATEX and plain TEX and pure pgf
\usepgflibrary[decorations.pathreplacing] % ConTEXt and pure pgf
\usetikzlibrary{decorations.pathreplacing} % LATEX and plain TEX when using Tik Z
\usetikzlibrary[decorations.pathreplacing]
\usepackage{stackrel}
\usepackage{xspace}

\renewcommand{\bar}[1]{\,\overline{\!#1\!}\,}

\newcommand{\DA}{\ang{A}}

\newcommand{\DB}{\ang{B}}

\newcommand{\DAbar}{\ang{\bar{A}}}
\newcommand{\DBbar}{\ang{\bar{B}}}

\newcommand{\DLbar}{\ang{\bar{L}}}
\newcommand{\BA}{[A]}
\newcommand{\BB}{[B]}

\newcommand{\bbO}{\mathbb{O}}

\newcommand{\ABB}{A\mspace{-0.3mu}B\bar{B}}
\newcommand{\ABBL}{A\mspace{-0.3mu}B\bar{B}\bar{L}}
\newcommand{\ABBA}{A\mspace{-0.3mu}B\bar{B}\bar{A}}
\newcommand{\AEEA}{A\mspace{-0.3mu}E\bar{E}\bar{A}}
\newcommand{\AEEL}{\bar{A}\mspace{-0.3mu}E\bar{E}L}
\newcommand{\AEE}{\bar{A}\mspace{-0.3mu}E\bar{E}}

\newcommand{\Bb}{B\mspace{-0.3mu}\bar{B}}
\newcommand{\EE}{E\mspace{-0.3mu}\bar{E}}
\newcommand{\PNL}{A\mspace{-0.3mu}\bar{A}}
\newcommand{\BBA}{\bar{A}\mspace{-0.3mu}B\bar{B}}
\newcommand{\Ba}{\bar{A}\mspace{-0.3mu}B}
\newcommand{\onlyA}{A}
\newcommand{\onlyAbar}{\bar{A}}
\newcommand{\onlyB}{B}
\newcommand{\onlyBbar}{\bar{B}}
\newcommand{\onlyE}{E}
\newcommand{\onlyEbar}{\bar{E}}
\newcommand{\onlyL}{L}
\newcommand{\onlyLbar}{\bar{L}}

\newcommand{\tABB}{\texorpdfstring{$\ABB$}{ABB}\xspace}
\newcommand{\tABBA}{\texorpdfstring{$\ABBA$}{ABBA}\xspace}
\newcommand{\tABBL}{\texorpdfstring{$\ABBL$}{ABBL}\xspace}
\newcommand{\tAEEL}{\texorpdfstring{$\AEEL$}{AEEL}\xspace}

\newcommand{\Inf}{\cI\mit{nf}}
\newcommand{\prop}{\cP\mit{rop}}
\newcommand{\closure}{\cC\mit{l}}
\newcommand{\eclosure}{\cC\mit{l}^+}
\newcommand{\witset}{\cW\mit{it}}
\newcommand{\futwit}{\cF\mit{ut}\witset}
\newcommand{\pastwit}{\cP\mit{ast}\witset}
\newcommand{\type}{\cT\mit{ype}}
\newcommand{\req}{\cR\mit{eq}}
\newcommand{\obs}{\cO\mit{bs}}
\newcommand{\shading}{\cS\mit{hading}}

\newcommand{\labeledrightarrow}[1]{\overset{\text{\raisebox{-0.1ex}[0ex][-0.1ex]{$_{#1\,}$}}}{\longrightarrow}}

\newcommand{\locallabeledrightarrow}[1]{\!\overset{\text{\raisebox{-0.1ex}[0ex][-0.1ex]{$_{#1\,}$}}}{\,\longmapsto}}

\newcommand{\dep}[1]{\,\text{\raisebox{-0.2ex}{$\labeledrightarrow{#1}$}}\,}

\newcommand{\localdep}[1]{\,\text{\raisebox{-0.2ex}{$\locallabeledrightarrow{#1}$}}\,}

\newcommand{\hs}{{\rm HS}}     % Halpern Shoham Logic
\newcommand{\pnl}{{\rm PNL}}     % PNL
\newcommand{\finishes}{\langle E \rangle}
\newcommand{\during}{\langle D \rangle}
\newcommand{\starts}{\langle B \rangle}
\newcommand{\overlaps}{\langle O \rangle}
\newcommand{\meets}{\langle A \rangle}

\newcommand{\xop}{\langle X \rangle}
\newcommand{\Txop}{\langle \overline{X} \rangle}

\newcommand{\later}{\langle L \rangle}

\title{Begin, After, and Later: a Maximal Decidable Interval Temporal Logic%
\thanks{This research was partly supported by the EU project FP7-ICT-223844 CON4COORD, by the Spanish-South African project HS2008-0006, by the Spanish MEC project TIN2009-14372-C03-01, and by the Italian GNCS project "Logics, Automata, and Games for the formal verification of complex systems".}}

\author{
Davide Bresolin
\institute{University of Verona\\ Verona, Italy\\{\em davide.bresolin@univr.it}}
\and
Pietro Sala
\institute{University of Verona\\ Verona, Italy\\{\em pietro.sala@univr.it}}
\and
Guido Sciavicco
\institute{University of Murcia\\
Murcia, Spain\\{\em guido@um.es}}
}
\def\titlerunning{Begin, After, and Later: a Maximal Decidable Interval Temporal Logic}
\def\authorrunning{D. Bresolin, P. Sala, \& G. Sciavicco}
\begin{document}
\maketitle

\begin{abstract}
Interval temporal logics (ITLs) are logics for reasoning about temporal
statements expressed over intervals, i.e., periods of time.
The most famous ITL studied so far is Halpern and Shoham's {\rm HS},
which is the logic of the thirteen Allen's interval relations.
Unfortunately, {\rm HS} and most of its fragments have an undecidable
satisfiability problem. This discouraged the research in this area until
recently, when a number non-trivial decidable ITLs have been
discovered. This paper is a contribution towards the complete classification of all
different fragments of {\rm HS}. We consider different combinations of the
interval relations {\em begins} ($\onlyB$), {\em after} ($\onlyA$), {\em later} ($\onlyL$) and their inverses $\onlyAbar$, $\onlyBbar$ and $\onlyLbar$. We know from previous works that the combination $\ABBA$ is decidable only when finite domains are considered (and undecidable elsewhere), and that $\ABB$ is decidable over the natural numbers. We extend these results by showing that decidability of $\ABB$ can be further extended to
capture the language $\ABBL$, which lies in between $\ABB$ and $\ABBA$, and that turns out to be maximal w.r.t decidability over strongly discrete linear orders (e.g. finite orders, the naturals, the integers). We also prove that the proposed decision procedure is optimal with respect to the EXPSPACE complexity class.
\end{abstract}

\section{Introduction}

Interval temporal logics (ITLs) are logics for reasoning about temporal statements expressed over intervals instead of points. The most famous ITL studied so far is probably Halpern and Shoham's \hs~\cite{interval_modal_logic}, which is the logic of (the thirteen) Allen's interval relations between intervals~\cite{interval_relations}. It features a modal operator for each relation, that is {\em after} ($\meets$) (also called {\em meets}), {\em begins} ($\starts$), {\em ends} ($\finishes$), {\em overlaps} ($\overlaps$), {\em during} ($\during$), {\em later} ($\later$), and their inverses (denoted by $\Txop$, where $\xop$ is a modal operator), although some of them are definable in terms of others. Since~\hs\ is undecidable when interpreted over almost all interesting classes of linearly ordered sets, it is natural to ask whether there exist decidable fragments of it, and how the properties of the underlying linearly ordered domain can influence its decidable/undecidable status. In the literature, the classes of linear orderings that have received more attention are i) the class of all linearly ordered sets, ii) the set of all discrete linearly ordered sets, iii) the class of all dense linearly ordered sets. In the second case one can also distinguish among {\em strong} discreteness (i.e., $\mathbb N,\mathbb Z$-like), and {\em weak} discreteness (which allows non-standard models such as $\mathbb N+\mathbb Z$). In recent years, a number of papers have been published in which new, sometimes unexpected, decidable and undecidable fragments are presented. Among them, we mention the fragment~\texorpdfstring{$\PNL$}{AA}, also known as~\pnl, presented in~\cite{pnl_logics}, and studied also in~\cite{pnl_expressiveness}, which is decidable over all interesting classes of models; and the fragment~\tABB (and, by symmetry,~\texorpdfstring{$\AEE$}{AEE}) which is decidable when interpreted over natural numbers~\cite{abb_natural}. Interestingly enough, the extension~\tABBA (and~\texorpdfstring{$\AEEA$}{AEEA}) turns out to be decidable only when finite models are considered, and undecidable as soon as an infinite ascending (resp., descending) chain is admitted in the model~\cite{abba_finite}. Other interesting fragments are~\texorpdfstring{$\Bb$}{BB} and~\texorpdfstring{$\EE$}{EE}, that are decidable in most cases~\cite{roadmap_intervals}, while any other combination of the four operators $\onlyB$, $\onlyBbar$, $\onlyE$, and $\onlyEbar$ immediately leads to undecidability~\cite{halpern_shoham_fragments}. Other combinations such as~\texorpdfstring{$\BBA$}{BBA}, and the simpler~\texorpdfstring{$\Ba$}{BA}, though, remain still uncovered.

\medskip

In this paper, we present another piece of this complicated puzzle by considering also the Allen's relation {\em later}, that captures any interval starting at some point after the ending point of the current interval, and it can be defined as $\meets\meets$, and the inverse relation \emph{before}. We will show that the logic~\tABBL (and the symmetric logic \tAEEL) is decidable and EXPSPACE-complete when interpreted over strongly discrete linear orders. It is worth emphasizing that adding any other non-definable Allen's relation to \tABBL and to \tAEEL leads to undecidability over all considered structures, with the exception of $\onlyAbar$ and $\onlyA$, respectively, which keep decidability only when finite models are considered (and cause undecidability over infinite models). Hence, our results shows also that \tABBL and \tAEEL are maximal fragments of \hs\ with respect to decidability in the class of all strongly discrete linear orders.

\medskip

The structure of this paper is as follows. In Section~\ref{sec:logic} we introduce syntax and semantics of our logic. In Section~\ref{sec:decidability}, we deal with the decidability of the satisfiability problem over finite and infinite structures, while in Section~\ref{sec:completeness} we discuss its complexity.  Finally, in Section~\ref{sec:conclusions} we draw some conclusions and outline future research directions.

\section{The interval temporal logic \tABBL}\label{sec:logic}
	\input{logic}

\section{Deciding the satisfiability problem for \tABBL}\label{sec:decidability}
	\input{smallmodel}

	\input{smallmodel-infinite}

\section{Complexity bounds to the satisfiability problem for \tABBL}\label{sec:completeness}
	\input{complexity}

\section{Conclusions}\label{sec:conclusions}
	\input{conclusions}

\bibliographystyle{eptcs} % or whatever you prefer
\bibliography{biblio}

%\newpage
%\input{appendix}

\end{document}

%% file: generic.tex
\usepackage[english]{babel}
\usepackage[latin1]{inputenc}
\usepackage{amsthm,amsmath,amssymb,amsfonts,mathrsfs,latexsym,stmaryrd}
\usepackage{array}
\usepackage{url}
\usepackage[newitem,newenum,neverdecrease]{paralist}
\usepackage{epsfig}
\usepackage{ifpdf}
\usepackage{color}
\definecolor{darkcyan}{rgb}{0,0.55,0.55}
\ifpdf
\else
\fi

\newenvironment{dotlist}{\begin{compactitem}}{\end{compactitem}}
\newenvironment{numlist}{\begin{compactenum}}{\end{compactenum}}

\newcommand{\refcond}[1]{\hyperref[#1]{C\ref*{#1}}}
\newcommand{\refprop}[1]{\hyperref[#1]{P\ref*{#1}}}
\newcommand{\refaxiom}[1]{\hyperref[#1]{A\ref*{#1}}}

\newcommand{\bbI}{\mathbb{I}}

\newcommand{\bbN}{\mathbb{N}}

\newcommand{\bbZ}{\mathbb{Z}}
\newcommand{\bbP}{\mathbb{P}}

\newcommand{\cA}{\mathcal{A}}

\newcommand{\cC}{\mathcal{C}}

\newcommand{\cF}{\mathcal{F}}
\newcommand{\cG}{\mathcal{G}}

\newcommand{\cI}{\mathcal{I}}

\newcommand{\cL}{\mathcal{L}}

\newcommand{\cO}{\mathcal{O}}
\newcommand{\cP}{\mathcal{P}}

\newcommand{\cR}{\mathcal{R}}
\newcommand{\cS}{\mathcal{S}}
\newcommand{\cT}{\mathcal{T}}

\newcommand{\cW}{\mathcal{W}}

\newcommand{\sP}{\mathscr{P}}

\newcommand{\s}[1]{\vspace{#1mm}}

\providecommand{\mit}{\mathit}
\renewcommand{\mit}{\mathit}

\makeatletter
\def\shortrightarrowfill@{\arrowfill@\relbar\relbar\shortrightarrow}
\newcommand{\ort}{\mathpalette{\overarrow@\shortrightarrowfill@}}
\makeatother
\makeatletter
\def\shortleftarrowfill@{\arrowfill@\relbar\relbar\shortleftarrow}
\newcommand{\olft}{\mathpalette{\overarrow@\shortleftarrowfill@}}
\makeatother
\makeatletter
\def\shortleftrightarrowfill@{\arrowfill@\relbar\relbar\leftrightarrow}
\newcommand{\olftrt}{\mathpalette{\overarrow@\shortleftrightarrowfill@}}
\makeatother

\newcommand{\nin}{\not\in}
\newcommand{\sat}{\vDash}

\newcommand{\et}{\;\wedge\;}
\newcommand{\vel}{\;\vee\;}
\newcommand{\then}{\;\rightarrow\;}

\newcommand{\ang}[1]{{\langle #1 \rangle}}

\newcommand{\len}[1]{{\lvert #1 \rvert}}

\newcommand{\img}{{\cI\mspace{-2mu}\mit{mg}}}

\renewcommand{\tt}[3]{\text{\small $\biggl(\begin{smallmatrix} #1 \\ #2 \\ #3 \end{smallmatrix}\biggr)$}}

\newcounter{theorem}

\newcommand{\theoremlike}[1]{\par\medskip\penalty-250\refstepcounter{theorem}{\bfseries\scshape\noindent#1 \thetheorem.}\slshape}
\newenvironment{theorem}{\theoremlike{Theorem}}{\par\medskip}

\newenvironment{proposition}{\theoremlike{Proposition}}{\par\medskip}
\newenvironment{lemma}{\theoremlike{Lemma}}{\par\medskip}
\newenvironment{definition}{\theoremlike{Definition}}{\par\medskip}

{\par\medskip}
{\par\medskip}

%% file: fixpseudocode.tex
\usepackage{pseudocode}
\usepackage{chngcntr}

\counterwithout{pseudocode}{section}
\newcounter{newpseudonum}[pseudocode]

\ifpdf
  \providecommand{\refline}[1]{\hyperref[#1]{(\ref*{#1})}}
\else
  \providecommand{\refline}[1]{\ref*{#1}}
\fi

\renewcommand{\ELSE}{\\\pcodetab{1}\mbox{ \bfseries \makebox[0pt][l]{else}\phantom{then} }}
\renewcommand{\RETURN}[1]{\ifthenelse{\equal{#1}{} }{\mbox{\bfseries return}}{\mbox{\bfseries return}#1}}
\newcommand{\FUNCTION}[2]{\mbox{\bfseries proc }\mbox{\textsc{#1}}\left(\ensuremath{#2}\right)\\}
\newcommand{\ENDFUNCTION}{}

\newlength{\pcodewidth}
\newenvironment{code}[1]{
\begin{Sbox}
\!\!\begin{minipage}{#1}%\pcodewidth}
\bfseries
\noindent
\scriptsize
$$
\begin{array}{@{\hspace*{1ex}}lr@{}}
}{
\end{array}
$$
\end{minipage}\vspace{-2mm}
\end{Sbox}
\shadowbox{\TheSbox}{}
}

%% file: logic.tex
In this section, we briefly introduce syntax and semantics of the logic $\ABBL$, along with the basic notions of atom, type, and dependency. We conclude the section by providing an alternative interpretation of $\ABBL$ over labeled grid-like structures.

\subsection{Syntax and semantics}\label{subsec:syntax}

The logic \tABBL features four modal operators $\DA$, $\DB$, $\DBbar$ and $\DLbar$, and it is interpreted in interval temporal structures over a strongly discrete linear order endowed with the four Allen's relations $A$ (``meets''), $B$ (``begins''), $\bar{B}$ (``begun by'') and $\bar{L}$ (``before'').  We recall that a linear order $\bbO=\ang{O,<}$ is \emph{strongly discrete} if and only if there are only finitely many points between any pair of points $x < y \in O$. Example of strongly discrete linear orders are all finite linear orders, and the sets $\bbN$ and $\bbZ$.

%\medskip

Given a set $\prop$ of propositional variables, formulas of $\ABBL$ are built up from $\prop$ using
the boolean connectives $\neg$ and $\vel$ and the unary modal operators $\DA$, $\DB$, $\DBbar$, $\DLbar$.
As usual, we shall take advantage of shorthands like $\varphi_1\et\varphi_2=\neg(\neg\varphi_1 \vel
\neg\varphi_2)$, $\BA\varphi=\neg\DA\neg\varphi$, $\BB\varphi=\neg\DB\neg\varphi$, etc. Hereafter,
we denote by $\len{\varphi}$ the size of $\varphi$. Given any strongly discrete linear order $\bbO=\ang{O,<}$  we  define $\bbI_\bbO$ as the set of all  closed intervals $[x,y]$, with $x,y\in O$ and $x < y$. For any pair of intervals $[x,y], [x',y'] \in \bbI_\bbO$, the Allen's relations ``meets'' $A$, ``begins'' $B$,  ``begun by'' $\bar{B}$, and ``before'' $\bar{L}$  are defined as follows:
\begin{dotlist}
  \item {\bf ``meets'' relation:} $[x,y] \;A\; [x',y']$ iff $y=x'$;
  \item {\bf ``begins'' relation:} $[x,y] \;B\; [x',y']$ iff $x=x'$ and $y'<y$;
  \item {\bf ``begun by'' relation:} $[x,y] \;\bar{B}\; [x',y']$ iff $x=x'$ and $y<y'$;
    \item {\bf ``before'' relation:} $[x,y] \;\bar{L}\; [x',y']$ iff  $y'<x$.
\end{dotlist}
Given an \emph{interval structure} $\cS=(\bbI_\bbO,A,B,\bar{B},\bar{L},\sigma)$, where $\sigma:\bbI_\bbO\then\sP(\prop)$
is a labeling function that maps intervals in $\bbI_\bbO$ to sets of propositional variables, and an initial
interval $I=[x,y]$, we define the semantics of an $\ABBL$ formula as follows:
\begin{dotlist}
  \item $\cS,I \sat a$ iff $a\in\sigma(I)$, for any $a\in\prop$;
  \item $\cS,I \sat \neg\varphi$ iff $\cS,I \not\sat\varphi$;
  %\item $\cS,I \sat \pi$ iff $x=y$;
  \item $\cS,I \sat \varphi_1 \vel \varphi_2$ iff $\cS,I \sat \varphi_1$ or $\cS,I \sat \varphi_2$;
  \item for every relation $R\in\{A,B,\bar{B},\bar{L}\}$, $\cS,I \sat \ang{R}\varphi$ iff there is an interval
        $J\in\bbI_\bbO$ such that $I \;R\; J$ and $\cS,J \sat \varphi$.
\end{dotlist}
Given an interval structure $\cS$ and a formula $\varphi$, we say that $\cS$ \emph{satisfies} $\varphi$
(and hence $\varphi$ is \emph{satisfiable}) if there is an interval $I$ in $\cS$ such that $\cS,I\sat\varphi$.
Accordingly, we define the \emph{satisfiability problem} for $\ABBL$ as the problem of establishing whether a
given $\ABBL$-formula $\varphi$ is satisfiable.

%\medskip

As we have recalled in the Introduction, we have that $\cS,I \sat \DLbar\varphi$ iff $\cS,I \sat \DAbar\DAbar\varphi$, and thus that $\DLbar$ is definable in the language of~\tABBA. As a direct consequence of the decidability and complexity results proved in this paper, we have that the converse it is not true.
Moreover, it is easy to see that the operator $\DLbar$ cannot be defined in the language of~\texorpdfstring{$\ABB$}{ABB}: the modal operators $\DA$, $\DB$ and $\DBbar$ allow the language to see only intervals whose endpoints are greater or equals to the endpoints of the interval were a formula is interpreted. Hence, the logic \tABBL is strictly more expressive than \tABB and strictly less expressive than \tABBA.

%\medskip
%
%It is worth mentioning that the fragment \texorpdfstring{$\onlyA$}{A} corresponds to the so-called Non-strict Right Propositional Neighborhood Logic (RPNL)~\cite{tableau_for_right_pnl}. Originally, the modal operator \emph{meets} was denoted by $\Diamond_r$, while the symbol $\DA$ was reserved for the {\em strict} version of RPNL, where point-intervals are excluded from the models. With the help of $\DB$, we can define the strict version of the operator \emph{meets} in the language of ~\texorpdfstring{$\ABB$}{ABB} (and, thus, of~\texorpdfstring{$\ABBL$}{ABB}), as
%%
%$$
%\DA(\psi \wedge \DB\top).
%$$

%\noindent The modal constant $\pi$, that captures exactly point-intervals, can be defined as $\BB\bot$.

\subsection{Atoms, types, and dependencies}\label{subsec:types}

Let $\cS=(\bbI_\bbO,A,B,\bar{B},\bar{L},\sigma)$ be an interval structure that satisfies the $\ABBL$-formula $\varphi$.
In the sequel, we relate intervals in $\cS$ with respect to the set of sub-formulas of $\varphi$
they satisfy. To do that, we introduce the key notions of $\varphi$-{\em atom} and $\varphi$-{\em type}.

%\medskip

First of all, we define the \emph{closure} $\closure(\varphi)$ of $\varphi$ as the set of all
sub-formulas of $\varphi$ and of their negations (we identify $\neg\neg\alpha$ with $\alpha$,
$\neg\DA\alpha$ with $\BA\neg\alpha$, etc.). For technical reasons, we also introduce the
\emph{extended closure} $\eclosure(\varphi)$, which is defined as the set of all formulas
in $\closure(\varphi)$ plus all formulas of the forms $\ang{R}\alpha$ and $\neg\ang{R}\alpha$,
with $R\in\{A,B,\bar{B}, \bar{L}\}$ and $\alpha\in\closure(\varphi)$. A \emph{$\varphi$-atom} is any non-empty set $F\subseteq\eclosure(\varphi)$ such that (i) for every
$\alpha\in\eclosure(\varphi)$, we have $\alpha\in F$ iff $\neg\alpha\nin F$ and (ii) for every
$\gamma=\alpha\vel\beta\in\eclosure(\varphi)$, we have $\gamma\in F$ iff $\alpha\in F$ or $\beta\in F$
(intuitively, a \emph{$\varphi$-atom} is a maximal {\sl locally consistent} set of formulas chosen
from $\eclosure(\varphi)$). Note that the cardinalities of both sets $\closure(\varphi)$ and
$\eclosure(\varphi)$ are {\sl linear} in the number $\len{\varphi}$ of sub-formulas of $\varphi$,
while the number of $\varphi$-atoms is {\sl at most exponential} in $\len{\varphi}$ (precisely,
we have $\len{\closure(\varphi)}=2\len{\varphi}$, $\len{\eclosure(\varphi)}=18\len{\varphi}$, and
there are at most $2^{8\len{\varphi}}$ distinct atoms). We define $\cA_\varphi$ as the set of all possible atoms that can be built over $\eclosure(\varphi)$.

%\medskip

We associate with each interval $I\in\cS$ the set of all formulas $\alpha\in\eclosure(\varphi)$
such that $\cS,I\sat\alpha$. Such a set is called \emph{$\varphi$-type} of $I$ and it is denoted by
$\type_\cS(I)$. We have that every $\varphi$-type is a $\varphi$-atom, but not vice versa. Hereafter,
we shall omit the argument $\varphi$, thus calling a $\varphi$-atom (resp., a $\varphi$-type) simply
an atom (resp., a type). Given an atom $F$, we denote by $\obs(F)$ the set of all \emph{observable} of $F$, namely, the formulas $\alpha\in\closure(\varphi)$ such that $\alpha\in F$. Similarly, given an atom $F$ and
a relation $R\in\{A,B,\bar{B},\bar{L}\}$, we denote by $\req_R(F)$ the set of all \emph{$R$-requests}
of $F$, namely, the formulas $\alpha\in\closure(\varphi)$ such that $\ang{R}\alpha\in F$. Taking
advantage of the above sets, we can define the following three relations between  two atoms $F$
and $G$:
$$
\begin{array}{rcl}
  F \dep{A} G
  &\;\quad\text{iff}\quad\quad&
  \req_A(F)                              \;=\;          \obs(G) \,\cup\, \req_B(G) \,\cup\, \req_{\bar{B}}(G)
\s1 \\
  F \dep{B} G
  &\;\quad\text{iff}\quad\quad&
  \begin{cases}
    \obs(F) \,\cup\, \req_{\bar{B}}(F)   \;\subseteq\;  \req_{\bar{B}}(G)
                                         \;\subseteq\;  \obs(F) \,\cup\, \req_{\bar{B}}(F) \,\cup\, \req_B(F)    \s1 \\
    \obs(G) \,\cup\, \req_B(G)           \;\subseteq\;  \req_B(F)
                                         \;\subseteq\;  \obs(G) \,\cup\, \req_B(G) \,\cup\, \req_{\bar{B}}(G)			\s1 \\
    \req_{\bar{L}}(F) = \req_{\bar{L}}(G).
  \end{cases}\s2\\
  F \dep{\bar{L}} G
  &\;\quad\text{iff}\quad\quad&
  \obs(G) \cup \req_{\bar{L}}(G)     \;\subseteq\;          \req_{\bar{L}}(F)    \s2 \\
\end{array}
$$
Note that the relations $\dep{B}{}$ and $\dep{\bar{L}}{}$ are transitive, while $\dep{A}{}$ is not. Moreover, all $\dep{A}{}$,
 $\dep{B}{}$ and $\dep{\bar{L}}{}$  satisfy a \emph{view-to-type dependency}, namely, for every pair of intervals $I,J$ in
$\cS$, we have that
$$
\begin{array}{rcl}
  I \;A\; J    &\;\quad\text{implies}\quad\quad&   \type_\cS(I) \,\dep{A}{}\, \type_\cS(J)     \s1 \\
  I \;B\; J    &\;\quad\text{implies}\quad\quad&   \type_\cS(I) \,\dep{B}{}\, \type_\cS(J)\s1 \\
  I \;\bar{L}\; J    &\;\quad\text{implies}\quad\quad&   \type_\cS(I) \,\dep{\bar{L}}{}\, \type_\cS(J).
\end{array}
$$

\subsection{Compass structures}\label{subsec:compass}

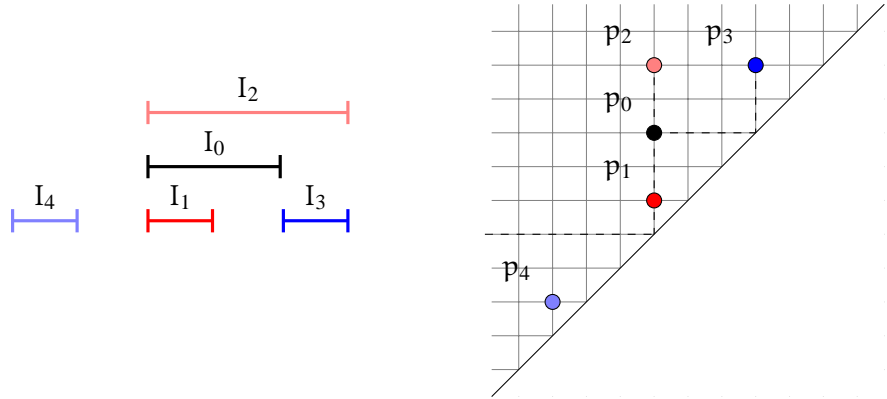
\begin{figure}[!!t]
\centering
\begin{tikzpicture}[scale=0.9]
\draw[very thick,|-|] (0,0) -- (2,0)node[pos=0.5, above] {$I_0$};
\draw[very thick,|-|,red] (0,-0.8) -- (1,-0.8)node[pos=0.5, above=0.001cm,black] {$I_1$};
\draw[very thick,|-|,red!50] (0,0.8) -- (3,0.8)node[pos=0.5, above=0.001cm,black] {$I_2$};
\draw[very thick,|-|,blue] (2,-0.8) -- (3,-0.8)node[pos=0.5, above=0.001cm,black] {$I_3$};
\draw[very thick,|-|,blue!50] (-1,-0.8) -- (-2,-0.8)node[pos=0.5, above=0.001cm,black] {$I_4$};
\pgftransformshift{\pgfpoint{8cm}{-0.5cm}}
\draw[step=0.5cm,gray,very thin] (-2.9,-2.9) grid (2.9,2.9);
\fill[color=white] (-2.9,-2.9) -- (2.9,2.9) -- (2.9,-2.9);
\draw (-2.9,-2.9) -- (2.9,2.9);
\draw[dashed] (-0.5,1) -- (-0.5, -0.5);
\draw[dashed] (-0.5,1) -- (-0.5, 2);
\draw[dashed] (-0.5,1) -- (1, 1);
\draw[dashed] (1,2) -- (1, 1);
\draw[dashed] (-3,-0.5) -- (-0.5, -0.5);
\node[shape=circle,draw=black,inner sep=2pt,fill=black, label={[label distance=0.1cm]above left:$p_0$}](A) at (-0.5,1) {};
\node[shape=circle,draw=black,inner sep=2pt,fill=red, label={[label distance=0.1cm]above left:$p_1$}](A) at (-0.5,0) {};
\node[shape=circle,draw=black,inner sep=2pt,fill=red!50, label={[label distance=0.1cm]above left:$p_2$}](A) at (-0.5,2) {};
\node[shape=circle,draw=black,inner sep=2pt,fill=blue, label={[label distance=0.1cm]above left:$p_3$}](A) at (1,2) {};
\node[shape=circle,draw=black,inner sep=2pt,fill=blue!50, label={[label distance=0.1cm]above left:$p_4$}](A) at (-2,-1.5) {};

\end{tikzpicture}
\caption{Correspondence between intervals and the points of a grid.}
\label{fig:compassstructure}
\end{figure}

The logic $\ABBL$ can be equivalently interpreted over grid-like structures (hereafter called {\em compass structures})
by exploiting the existence of a natural bijection between the intervals $I=[x,y]$ and the points $p=(x,y)$ of
an $O\times O$ grid such that $x < y$. As an example, in Fig.~\ref{fig:compassstructure} are shown five intervals
$I_0,...,I_4$, such that $I_0 \;B\; I_1$,  $I_0 \;\bar{B}\;I_2$, $I_0 \;A\; I_3$, and $I_0 \;\bar{L}\;I_4$,   together with the corresponding
points $p_0,...,p_4$ of a  grid (note that the four Allen's relations $A,B,\bar{B},\bar{L}$ between intervals
are mapped to the corresponding spatial relations between points; for the sake of readability, we name the latter
ones as the former ones).

\begin{definition}\label{def:compassstructure}
Given an $\ABBL$ formula $\varphi$, a (consistent and fulfilling) \emph{compass} ($\varphi$-)\emph{structure}
 is a pair $\cG=(\bbP_\bbO,\cL)$, where $\bbP_\bbO$ is the set of points of the form $p=(x,y)$, with $x,y \in O$ and
$ x<y$, and $\cL$ is function that maps any point $p\in\bbP_\bbO$ to a ($\varphi$-)atom $\cL(p)$ in such a way that:
\begin{dotlist}
  \item for every pair of points $p,q\in\bbP_\bbO$ and every relation $R\in\{A,B,\bar{L}\}$,
        if $p \;R\; q$ holds, then $\cL(p) \dep{R} \cL(q)$ follows ({\bf consistency});
  \item for every point $p\in\bbP_\bbO$, every relation $R\in\{A,B,\bar{B},\bar{L}\}$, and
        every formula $\alpha\in\req_R\bigl(\cL(p)\bigr)$, there is a point $q\in\bbP_\bbO$ such that
        $p \;R\; q$ and $\alpha\in\obs\bigl(\cL(q)\bigr)$ ({\bf fulfillment}).
\end{dotlist}
\end{definition}

\noindent
We say that a compass ($\varphi$-)structure $\cG=(\bbP_\bbO,\cL)$ \emph{features} a formula
$\alpha$ if there is a point $p\in\bbP_\bbO$ such that $\alpha \in \cL(p)$. The following
proposition implies that the satisfiability problem for $\ABBL$ is reducible to the problem
of deciding, for any given formula $\varphi$, whether there exists a $\varphi$-compass
structure featuring $\varphi$.

\begin{proposition}\label{prop:compassstructure}
An $\ABBL$-formula $\varphi$ is satisfied by some interval structure if and only if it is
featured by some ($\varphi$-)compass structure.
\end{proposition}

%% file: smallmodel.tex
In this section, we prove that the satisfiability problem for $\ABBL$ is decidable
by providing a ``small-model theorem'' for the satisfiable formulas of the logic.
For the sake of simplicity, we first show that the satisfiability problem for
$\ABBL$ interpreted over {\sl finite} interval structures is decidable and then we
generalize such a result to all (finite or infinite) interval structures based on strong discrete linear orders.

%\medskip

As a preliminary step, we introduce the key notions of shading, of witness set, and of compatibility between rows of a compass structure. Let $\cG=(\bbP_\bbO,\cL)$ be a compass structure and let $y \in O$. The \emph{shading of the row $y$ of $\cG$} is the set $\shading_\cG(y)=\bigl\{\cL(x,y) \,:\, x<y\bigr\}$, namely, the set of the atoms of all points in $\bbP_\bbO$ whose vertical coordinate has value $y$ (basically, we interpret different atoms as different colors).
%Clearly, for every pair of atoms $F$ and $F'$ in $\shading_\cG(y)$, we have that $\req_A(F)=\req_A(F')$.
A \emph{witness set for $y$} is any \emph{minimal} set $\witset(y) \subseteq \{(x_\psi,y_\psi) : x_\psi< y_\psi \wedge y_\psi>y\}$ that respects the following property:
\begin{compactitem}
	\item[\bf(WIT)] for every $\psi \in \closure(\varphi)$ that appears in the labeling of some point $(x',y')$ with $y' > y$, there exists a \emph{witness} $(x_\psi, y_\psi) \in \witset(y)$ such that
	\begin{compactenum}
		\item $\psi \in \cL(x_\psi, y_\psi)$ , and
		\item $y_\psi$ is minimal, that is, for all $(x',y')$ with $y < y' < y_\psi$, $\neg\psi \in \cL(x', y')$.
	\end{compactenum}
\end{compactitem}
\noindent
Since $\witset(y)$ is minimal we have that there is at most one distinct point for every $\psi \in \closure(\varphi)$
and thus $|\witset(y)|<|\closure(\varphi)|=2\cdot \len{\varphi}$. Intuitively, a witness set for a row $y$ is a set that contains, for every formula $\psi$ that occurs in some point above the row $y$, a witness $(x_\psi, y)$ for it, that is, a point that satisfies $\psi$ at the minimum possible distance from the row $y$. The notion of shading and of witness set allow us to determine whether two rows are \emph{compatible} or not.

Let $P\subseteq\bbP_\bbO$ a set of points and $\bar{y}$ be a coordinate we define $\pi_{\bar{y}}(P)=\{ x : (x,y)\in P \wedge x<{\bar{y}} ) \}$, the set of all $x$-coordinate belonging to points in $P$ which are smaller than $\bar{y}$.

\begin{definition}\label{def:compatible-rows}
Given a compass structure $\cG$ and two rows $y_0 < y_1$, we say that $y_0$ and $y_1$ are \emph{compatible} if and only if the following properties holds:
\begin{compactenum}
	\item $\shading_\cG(y_0) = \shading_\cG(y_1)$;
	\item $\cL(y_0-1, y_0) = \cL(y_1 - 1, y_1)$;
	\item there exists a witness set $\witset(y_1)$ for $y_1$ and an \emph{injective mapping function} $w:\pi_{y_1}(\witset(y_1))\mapsto \{x : x < y_0 \}$ s.t. $\cL(x,y_1)=\cL(w(x),y_0)$ for every $x \in \pi_{y_1}(\witset(y_1))$, that assigns a distinct $x$-coordinate on the row $y_0$ for every  witness $(x_\psi,y_\psi)$ in $\witset(y_1)$ with $x_\psi \leq y_1$.
\end{compactenum}
\end{definition}

\noindent In the following, we will show how the properties of compatible rows can be used to contract compass structures to smaller ones, first for finite models and then for infinite ones.

\subsection{A small-model theorem for finite structures}\label{subsec:finitecase}

Let $\varphi$ be an $\ABBL$ formula. It is easy to see that $\varphi$ is satisfiable over a finite model if and only if the formula $\varphi \vee \DBbar\varphi \vee \DA\varphi \vee \DA\DA\varphi$ is featured by the \emph{initial point} $(0,1)$ a {\sl finite} compass structure $\cG=(\bbP_\bbO,\cL)$. We prove that we can restrict our attention to compass structures $\cG=(\bbP_\bbO,\cL)$ with a number of points in $O$ bounded by a double exponential in $\len{\varphi}$. We start with the following lemma that proves two simple, but crucial, properties of the relations  $\dep{A}$, $\dep{B}$, and $\dep{\bar{L}}$. 

\begin{lemma}\label{lemma:entanglement}
Let $F,G,H$ be some atoms:
\begin{compactenum}
	\item if $F\dep{A}H$ and $G\dep{B}H$ hold, then $F\dep{A}G$ holds as well;
	\item if $F\dep{B}G$ and $G\dep{\bar{L}}H$ hold, then $F\dep{\bar{L}}H$ holds as well.
\end{compactenum}
\end{lemma}

\begin{proof}
	The proof for property 1 can be found in~\cite{abb_report}. As for property 2, we have that, by the definition of $\dep{B}$, if $F\dep{B}G$ then $\req_{\bar{L}}(F) = \req_{\bar{L}}(G)$. This implies that $\obs(H) \cup \req_{\bar{L}}(H) \subseteq \req_{\bar{L}}(F)$ and thus $F\dep{\bar{L}}H$ holds as well.
\end{proof}

The next lemma shows that, under suitable conditions, a given compass structure
$\cG$ may be reduced in length, preserving the existence of atoms featuring $\varphi$.

\begin{lemma}\label{lemma:contraction_finite}
Let $\cG$ be a finite compass structure  of size $N$ featuring $\varphi$ on the initial point $(0,1)$. If there exist two compatible rows $0<y_0<y_1<N$ in $\cG$, then there exists a compass structure $\cG'$ of size $N' = N - y_1 + y_0$ that features $\varphi$.
\end{lemma}

\begin{proof}
Suppose that $0<y_0<y_1<N$ are two compatible rows of $\cG$. By definition, we have that $\shading_\cG(y_0) = \shading_\cG(y_1)$, $\cL(y_0-1, y_0) = \cL(y_1 - 1, y_1)$, and there exists a witness set $\witset(y_1)$ for $y_1$ and an \emph{injective mapping function} $w: \pi_{y_1}(\witset(y_1))\mapsto \{x : x < y_0 \}$ that assigns a distinct $x$-coordinate on the row $y_0$ for every witness $(x_\psi,y_\psi)$ in $\witset(y_1)$ with $x_\psi \leq y_1$.
Then, we can define a function $f:\{0,...,y_0-1\}\mapsto\{0,...,y_1-1\}$ such that,
for every $0\le x<y_0$, $\cL(x,y_0)=\cL(f(x),y_1)$ and for every $(x_\psi,y_\psi)\in\witset(y_1)$ if $x_\psi<y_1$   then $f(w(x_\psi))=x_\psi$.

Let $k=y_1-y_0$, $N'=N-k$ ($<N$), $\bbO' = \langle \{0, \ldots, N'-1\}, <\rangle$, and $\bbP_{\bbO'}$ be the correspondent portion of the grid. We extend $f$ to a function that maps points in $\bbP_{\bbO'}$ to points in $\bbP_\bbO$ as follows:
\begin{dotlist}
  \item if $p=(x,y)$, with $0\le x<y<y_0$, then we simply let $f(p)=p$;
  \item if $p=(x,y)$, with $0\le x<y_0\le y$, then we let $f(p)=(f(x),y+k)$;
  \item if $p=(x,y)$, with $y_0\le x<y$, then we let $f(p)=(x+k,y+k)$.
\end{dotlist}
We denote by $\cL'$ the labeling of $\bbP_{\bbO'}$ such that, for every point $p\in\bbP_{\bbO'}$,
$\cL'(p)=\cL(f(p))$ and we denote by $\cG'$ the resulting structure $(\bbP_{\bbO'},\cL')$ (see Figure
\ref{fig:contraction}). We have to prove that $\cG'$ is a consistent and fulfilling
compass structure that features $\varphi$. First, we
show that $\cG'$ satisfies the consistency conditions for the relations $B$, $A$, and $\bar{L}$; then we show
that $\cG'$ satisfies the fulfillment conditions for the $\bar{B}$-, $B$-, $A$, and $\bar{L}$-requests;
finally, we show that $\cG'$ features $\varphi$.

\medskip\noindent\textsc{Consistency with relation $B$.\;\;}
Consider two points $p=(x,y)$ and $p'=(x',y')$ in $\cG'$ such that $p \;B\; p'$, i.e., $0\le x=x'<y'<y<N'$.
We prove that $\cL'(p) \dep{B} \cL'(p')$ by distinguishing among the following three cases (note that
exactly one of such cases holds):
\begin{numlist}
  \item $y<y_0$ and $y'<y_0$,
  \item $y\ge y_0$ and $y'\ge y_0$,
  \item $y\ge y_0$ and $y'<y_0$.
\end{numlist}

If $y<y_0$ and $y'<y_0$, then, by construction, we have $f(p)=p$ and $f(p')=p'$. Since $\cG$
is a (consistent) compass structure, we immediately obtain $\cL'(p)=\cL(p) \dep{B} \cL(p')=\cL'(p')$.

If $y\ge y_0$ and $y\ge y_0$, then, by construction, we have either $f(p)=(f(x),y+k)$ or $f(p)=(x+k,y+k)$,
depending on whether $x<y_0$ or $x\ge y_0$. Similarly, we have either $f(p')=(f(x'),y'+k)=(f(x),y'+k)$ or
$f(p')=(x'+k,y'+k)=(x+k,y'+k)$. This implies $f(p) \;B\; f(p')$ and thus, since $\cG$ is a (consistent)
compass structure, we have $\cL'(p)=\cL(f(p))$ $\dep{B}$ $\cL(f(p'))=\cL'(p')$.

If $y\ge y_0$ and $y'<y_0$, then, since $x<y'<y_0$, we have by construction $f(p)=(f(x),y+k)$ and $f(p')=p'$.
Moreover, if we consider the point $p''=(x,y_0)$ in $\cG'$, we easily see that (i) $f(p'')=(f(x),y_1)$,
(ii) $f(p) \;B\; f(p'')$ (whence $\cL(f(p)) \dep{B} \cL(f(p''))$), (iii) $\cL(f(p''))=\cL(p'')$, and
(iv) $p'' \;B\; p'$ (whence $\cL(p'') \dep{B} \cL(p')$). It thus follows that
$\cL'(p)=\cL(f(p)) \dep{B} \cL(f(p''))$ $=\cL(p'')$ $\dep{B}$ $\cL(p')=\cL(f(p'))=\cL'(p')$. Finally,
by exploiting the transitivity of the relation $\dep{B}$, we obtain $\cL'(p) \dep{B} \cL'(p')$.

\medskip\noindent
\textsc{Consistency with relation $A$.\;\;}
Consider two points $p=(x,y)$ and $p'=(x',y')$ such that $p \;A\; p'$, i.e., $0\le x<y=x'<y'<N'$.
We define $p''=(y,y+1)$ in such a way that $p \;A\; p''$ and $p' \;B\; p''$ and we distinguish between
the following two cases:
\begin{numlist}
  \item $y\ge y_0$,
  \item $y<y_0$.
\end{numlist}

If $y\ge y_0$, then, by construction, we have $f(p) \;A\; f(p'')$. Since $\cG$ is a (consistent)
compass structure, it follows that $\cL'(p)=\cL(f(p))$ $\dep{A}$ $\cL(f(p''))=\cL'(p'')$.

If $y<y_0$,
then, by construction, we have $\cL(p'')=\cL(f(p''))$. Again, since $\cG$ is a (consistent) compass
structure, it follows that $\cL'(p)=\cL(f(p))=\cL(p)$ $\dep{A}$ $\cL(p'')=\cL(f(p''))=\cL'(p'')$.

In both cases we have $\cL'(p) \dep{A} \cL'(p'')$. Now, we recall that $p' \;B\; p''$ and that, by
previous arguments, $\cG'$ is consistent with the relation $B$. We thus have $\cL'(p') \dep{B} \cL'(p'')$.
Finally, by applying Lemma \ref{lemma:entanglement}, we obtain $\cL'(p) \dep{A} \cL'(p')$.

\begin{figure}[tbp]
\centering
\includegraphics[scale=0.8]{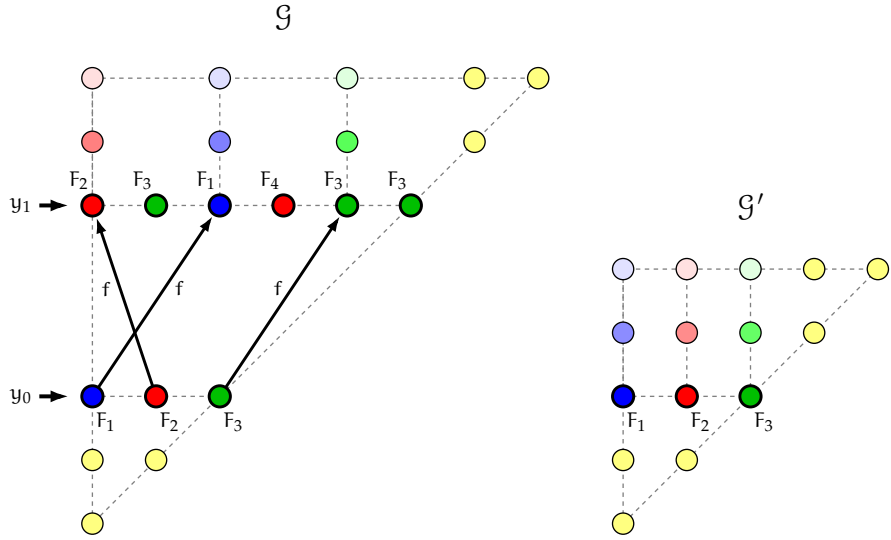}
\caption{Contraction $\cG'$ of a compass structure $\cG$.}
\label{fig:contraction}
\end{figure}

\medskip\noindent
\textsc{Consistency with relation $\bar{L}$.\;\;}
Consider two points $p=(x,y)$ and $p'=(x',y')$ in $\cG'$ such that $p \;\bar{L}\; p'$, i.e., $0\le x'<y'<x<y<N'$.
We prove that $\cL'(p) \dep{\bar{L}} \cL'(p')$ by distinguishing among the following three cases (note that
exactly one of such cases holds):
\begin{numlist}
  \item $y<y_0$ and $y'<y_0$,
  \item $y\ge y_0$ and $y'\ge y_0$,
  \item $y\ge y_0$ and $y'<y_0$.
\end{numlist}

If $y<y_0$ and $y'<y_0$, then, by construction, we have $f(p)=p$ and $f(p')=p'$. Since $\cG$
is a (consistent) compass structure, we immediately obtain $\cL'(p)=\cL(p) \dep{\bar{L}} \cL(p')=\cL'(p')$.

If $y\ge y_0$ and $y'\ge y_0$, then, by construction, we have either $f(p')=(f(x'),y'+k)$ or $f(p')=(x'+k,y'+k)$,
depending on whether $x'<y_0$ or $x'\ge y_0$. Since $y_0 \leq y' < x$, we have $f(p)=(x+k,y+k)$. This implies $f(p) \;\bar{L}\; f(p')$ and thus, since $\cG$ is a (consistent)
compass structure, we have $\cL'(p)=\cL(f(p))$ $\dep{\bar{L}}$ $\cL(f(p'))=\cL'(p')$.

If $y\ge y_0$ and $y'<y_0$, then, we have by construction that $f(p')=p'$ and either $f(p)=(x+k,y+k)$ or $f(p)=(f(x),y+k)$. In the former case we have that $f(p)\;\bar{L}\;f(p')$ and thus, since $\cG$ is a consistent compass structure, $\cL'(p)=\cL(f(p))\dep{\bar{L}}\cL(f(p'))=\cL'(p')$. In the latter case it is not necessarily true that $y' < f(x)$. Consider the points $p'' = (f(x), y_1)$ and $p''' = (x, y_0)$: by the definition of $f$, $\cL(p'') = \cL(p''')$. Moreover, we have that $f(p) B p''$ and $p''' \bar{L} f(p') = p'$.
Since $\cG$ is a consistent compass structure, this implies that $\cL'(p) = \cL(f(p)) \dep{B} \cL(p'') = \cL(p''') \dep{\bar{L}} \cL(f(p')) = \cL'(p')$. 
Finally, by applying Lemma \ref{lemma:entanglement}, we obtain $\cL'(p) \dep{\bar{L}} \cL'(p')$.

\medskip\noindent
\textsc{Fulfillment of $B$-requests.}\!
Consider a point $p=(x,y)$ in $\cG'$ and some $B$-request $\alpha\in\req_B\bigl(\cL'(p)\bigr)$ associated
with it. Since, by construction, $\alpha\in\req_B\bigl(\cL(f(p))\bigr)$ and $\cG$ is a (fulfilling)
compass structure, we know that $\cG$ contains a point $q'=(x',y')$ such that $f(p) \;B\; q'$ and
$\alpha\in\obs\bigl(\cL(q')\bigr)$. We prove that $\cG'$ contains a point $p'$ such that $p \;B\; p'$
and $\alpha\in\obs\bigl(\cL'(p')\bigr)$ by distinguishing among the following three cases (note that
exactly one of such cases holds):
\begin{numlist}
  \item $y<y_0$
  \item $y'\ge y_1$,
  \item $y\ge y_0$ and $y'<y_1$.
\end{numlist}

If $y<y_0$, then, by construction, we have $p=f(p)$ and $q'=f(q')$. Therefore, we simply
define $p'=q'$ in such a way that $p=f(p) \;B\; q'=p'$ and $\alpha\in\obs\bigl(\cL'(p')\bigr)$
($=\obs\bigl(\cL(f(p'))\bigr)=\obs\bigl(\cL(q')\bigr)$).

If $y'\ge y_1$, then, by construction, we have either $f(p)=(f(x),y+k)$ or $f(p)=(x+k,y+k)$,
depending on whether $x<y_0$ or $x\ge y_0$. We define $p'=(x,y'-k)$ in such a way that $p \;B\; p'$.
Moreover, we observe that either $f(p')=(f(x),y')$ or $f(p')=(x+k,y')$, depending on whether $x<y_0$
or $x\ge y_0$, and in both cases $f(p')=q'$ follows. This shows that $\alpha\in\obs\bigl(\cL'(p')\bigr)$
($=\obs\bigl(\cL(f(p')\bigr)=\obs\bigl(\cL(q')\bigr)$).

If $y\ge y_0$ and $y'<y_1$, then we define $\bar{p}=(x,y_0)$ and $\bar{q}=(x',y_1)$ and we observe
that $f(p) \;B\; \bar{q}$, $\bar{q} \;B\; q'$, and $f(\bar{p})=\bar{q}$. From $f(p) \;B\; \bar{q}$
and $\bar{q} \;B\; q'$, it follows that $\alpha\in\req_B\bigl(\cL(\bar{q})\bigr)$ and hence
$\alpha\in\req_B\bigl(\cL(\bar{p})\bigr)$.
Since $\cG$ is a (fulfilling) compass structure, we know that there is a point $p'$ such that
$\bar{p} \;B\; p'$ and $\alpha\in\obs\bigl(\cL(\bar{p}')\bigr)$. Moreover, since $\bar{p} \;B\; p'$,
we have $f(p')=p'$, from which we obtain $p \;B\; p'$ and $\alpha\in\obs\bigl(\cL(p')\bigr)$.

\medskip\noindent
\textsc{Fulfillment of $\bar{B}$-requests.}
The proof that $\cG'$ fulfills all $\bar{B}$-requests of its atoms is symmetric with respect to
the previous one.

\medskip\noindent
\textsc{Fulfillment of $A$-requests.\!}
Consider a point $p=(x,y)$ in $\cG'$ and some $A$-request\! $\alpha\in\req_A\bigl(\!\cL'(p)\!\bigr)$
associated with $p$ in $\cG'$. Since, by previous arguments, $\cG'$ fulfills all $\bar{B}$-requests
of its atoms, it is sufficient to prove that either $\alpha\in\obs\bigl(\cL'(p')\bigr)$ or
$\alpha\in\req_{\bar{B}}\bigl(\cL'(p')\bigr)$, where $p'=(y,y+1)$. This can be easily proved by
distinguishing among the three cases $y<y_0-1$, $y=y_0-1$, and $y\ge y_0$.

\medskip\noindent
\textsc{Fulfillment of $\bar{L}$-requests.\!}
Consider a point $p=(x,y)$ in $\cG'$ and some $\bar{L}$-request $\alpha\in\req_{\bar{L}}\bigl(\cL'(p)\bigr)$ associated with it. Since, by construction, $\alpha\in\req_{\bar{L}}\bigl(\cL(f(p))\bigr)$ and $\cG$ is a (fulfilling) compass structure, we know that $\cG$ contains a point $q'=(x',y')$ such that $f(p) \;\bar{L}\; q'$ and $\alpha\in\obs\bigl(\cL(q')\bigr)$. To simplify the proofs, we assume that $q'$ is \emph{minimal} with respect to the vertical coordinate, that is, for every other point $q''=(x'',y'')$ with $y'' < y'$, $\alpha\not\in\obs\bigl(\cL(q'')\bigr)$. We prove that $\cG'$ contains a point $p'$ such that $p \;\bar{L}\; p'$
and $\alpha\in\obs\bigl(\cL'(p')\bigr)$ by distinguishing among the following five cases (note that
exactly one of such cases holds):
\begin{numlist}
  \item $y\leq y_0$,
  \item $x < y_0$ and $y \geq y_0$, 
  \item $x \geq y_0$ and $y' < y_1$,
  \item $x \geq y_0$ and $y' = y_1$,
  \item $x \geq y_0$ and $y' > y_1$.
\end{numlist}

If $y<y_0$, then, by construction, we have $p=f(p)$ and $q'=f(q')$. Therefore, we simply
define $p'=q'$ in such a way that $p=f(p) \;\bar{L}\; q'=p'$ and $\alpha\in\obs\bigl(\cL'(p')\bigr)$
($=\obs\bigl(\cL(f(p'))\bigr)=\obs\bigl(\cL(q')\bigr)$).

If $x < y_0$ and $y \geq y_0$ then $f(p) = (f(x),y+k)$. Now, consider the point $p'' = (f(x),y_1)$: since $f(p) B p''$ and $\cG$ is a consistent compass structure, we have that $\req_{\bar{L}}(p'') = \req_{\bar{L}}(f(p))$. By definition of $f$, we have that $\cL(f(x),y_1) = \cL(x,y_0)$ and thus, since $\cG$ is fulfilling, there exists a point $p'=(x'',y'')$ such that $y'' < x$ and $\alpha \in\obs\bigl(\cL(p')\bigr)$. Since $f(p') = p'$, this shows that $\alpha\in\obs\bigl(\cL'(p')\bigr)$ as well.

If $x \geq y_0$ and $y' < y_1$ then $f(p) = (x + k, y  + k)$. Since $\cG$ is a consistent compass structure, we have that $\alpha \in \req_{\bar{L}}(\cL(y_1-1,y_1))$. By the definition of compatible rows, we have that $\cL(y_1-1,y_1) = \cL(y_0-1,y_0)$ and thus (by the minimality assumption) $y' < y_0$ and $q'=f(q')$. Therefore, we simply define $p'=q'$ in such a way that $p \;\bar{L}\; q'=p'$ and $\alpha\in\obs\bigl(\cL'(p')\bigr)$
($=\obs\bigl(\cL(f(p'))\bigr)=\obs\bigl(\cL(q')\bigr)$).
 
If $x \geq y_0$ and $y' = y_1$ then $\cL(q') \in \shading_\cG(y_1)$. By the definition of compatible rows, we have that $\shading_\cG(y_1) = \shading_\cG(y_0)$ and thus there must exists a point $q'' = (x'',y_0)$ such that $\cL(q') = \cL(q'')$ and $y_0 < y'$, against the hypothesis that $q'$ is a minimal point satisfying $\alpha$. Hence, this case cannot happen.

If $x \geq y_0$ and $y' > y_1$ then, by the minimality assumption on $q'$ we have that for every $y'' < y'$, $\alpha\not\in\obs\bigl(\cL(x'',y'')\bigr)$ for any $x'' < y''$. Hence, by the definition of witness set, we have that there exists a witness $(x_\alpha, y_\alpha) \in \witset(y_1)$ such that $\alpha\in\obs\bigl(\cL(x_\alpha,y_\alpha)\bigr)$ and $y_\alpha = y'$ (by the minimality assumption). If $x_\alpha \geq y_1$ then we define $p' = (x_\alpha-k, y_\alpha - k)$. Otherwise, $x_\alpha < y_1$ and by the definition of the mapping function $w$ and of the function $f$, we have that $f(w(x_\alpha)) = x_\alpha$: we define $p' = (w(x_\alpha), y' - k)$. In both cases we have that $f(p') = (x_\alpha, y_\alpha)$, $p \bar{L} p'$ and $\alpha\in\obs\bigl(\cL'(p')\bigr)$.

\medskip\noindent
\textsc{Featured formulas.\;\;}
Recall that, by previous assumptions, $\varphi\in\cL(0,1)$. Since our contraction procedure never changes the labelling of the initial point, $\varphi\in\cL'(0,1)$ as well.
\end{proof}

\medskip
On the grounds of the above result, we can provide a suitable upper bound for the length
of a minimal finite interval structure that satisfies $\varphi$, if there exists any.
This yields a straightforward, but inefficient, 2NEXPTIME algorithm that decides
whether a given $\ABBL$-formula $\varphi$ is satisfiable over finite interval structures.

\begin{theorem}\label{th:contraction_finite}
An $\ABBL$-formula $\varphi$ is satisfied by some finite interval structure iff it is
featured by some compass structure of length $N\le (8|\varphi| + 15)^{2^{32\len{\varphi}+56}} \cdot 2^{32\len{\varphi}+56}$ (i.e., double
exponential in $\len{\varphi}$).
\end{theorem}

\begin{proof}
Suppose that $\varphi$ is satisfied by a finite interval structure $\cS$, and let $\xi = \varphi \vee \DBbar\varphi \vee \DA\varphi \vee \DA\DA\varphi$. By Proposition \ref{prop:compassstructure}, there
is a compass structure $\cG$ that features $\xi$ on the initial point and has finite length $N$. By Lemma~\ref{lemma:contraction_finite}, we can assume without loss of generality that all rows of $\cG$ are pairwise incompatible. We recall from Section \ref{subsec:types} that $\cG$ contains at most $2^{8\len{\xi}}$ distinct
atoms. For every row $y$ of the compass structure and every atom $F \in \cA_\xi$, let $\#(F, y)$ be the cardinality of the set $\{(x,y) : x < y$ and $\cL(x,y) = F\}$. We associate to every row $y$ of the structure a \emph{characteristic function} $c_y : \cA_\xi \mapsto \bbN$ defined as follows:
\begin{eqnarray}
	c_y(F) & = & \begin{cases}
								\#(F, y)		&	 \#(F, y) \leq 2|\xi| \\
								2|\xi|			& \text{otherwise}
							\end{cases}
\end{eqnarray}

\noindent Since any witness set $\witset(y)$ contains at most $2|\xi|$ witnesses, it is easy to see that two rows $y_0$ and $y_1$ with the same characteristic function and such that $\cL(y_0-1,y_0) = \cL(y_1-1,y_1)$ are compatible. The number of possible characteristic functions is bounded by $(2|\xi|+1)^{2^{8\len{\xi}}}$, and thus $\cG$ cannot have more than $(2|\xi|+1)^{2^{8\len{\xi}}}\cdot 2^{8\len{\xi}}$ rows. 
Since $|\xi| = 4|\varphi| + 7$ we can conclude that $N \le (8|\varphi| + 15)^{2^{32\len{\varphi}+56}} \cdot 2^{32\len{\varphi}+56}$, and thus double exponential in $\len{\varphi}$.
\end{proof}

%% file: smallmodel-infinite.tex
\subsection{A small-model theorem for infinite structures}\label{subsec:infinitecase}

In general, compass structures that feature $\varphi$ may be infinite. Here, we prove
that, without loss of generality, we can restrict our attention to sufficiently ``regular''
infinite compass structures, which can be represented in double exponential space with respect
to $\len{\varphi}$. To do that, we introduce the notion of compass structure generator, that is, of a finite compass structure featuring $\varphi$ that can be extended to an infinite fulfilling one.

\begin{definition}\label{def:partialfulfilling}
We say that a finite compass structure $\cG=(\bbP_\bbO,\cL)$ of size $N$ is \emph{partially fulfilling} if
for every point $(x,y) \in \bbP_\bbO$ such that $y < N-1$, for every relation $R\in \{A,B,\bar{B},\bar{L}\}$, and for every formula $\psi \in \req_R(\cL(p))$, one of the following conditions hold:
  \begin{compactenum}
  \item there exists a point $p' \in \bbP_\bbO$ such that  $p\ R\ p'$ and $\psi \in \obs(\cL(p'))$	($\psi$ is fulfilled in $p'$),
  \item $R=\bar{B}$ and $\psi \in \req_{\bar{B}}(\cL(x,N-1))$,
  \item $R=A$ and $\psi \in \req_{\bar{B}}(\cL(y,N-1))$,
  \item $R=\bar{L}$ and $\psi \in \req_{\bar{L}}(\cL(0,1))$.
  \end{compactenum}
\end{definition}

\noindent Notice that all $B$-requests are fulfilled in a partially fulfilling compass structure and that $\bar{B}$, $A$, and $\bar{L}$ requests are either fulfilled or ``transferred to the border'' of the compass structure. Moreover, any substructure $\cG'$ of a fulfilling compass structure $\cG$ is partially fulfilling.

\begin{definition}\label{def:futurewitset}
Given a finite compass structure $\cG=(\bbP_\bbO,\cL)$ and a row $y$, a \emph{future witness set for $y$} is any \emph{minimal} set $\futwit(y) \subseteq \{x : x < y \}$ such that for every $F\in \shading_\cG(y)$ there exists a witness $x_F\in \futwit(y)$ that respects the following properties:
\begin{compactenum}
	\item $\cL(x_F,y)=F$,
	\item for every $\psi \in \req_{\bar{B}}(F)$ there exists a point $(x_F,y')\in \cG$ with $y'>y$ and $\psi \in \obs(\cL(x_F,y))$.
\end{compactenum}
\end{definition}

\noindent Since $\futwit(y)$ is minimal, we have that  for every $F \in \shading(y)$ there is exactly one witness $x_F$ in $\futwit(y)$. Hence, $|\futwit(y))|\leq 2^{8\len{\varphi}}$.   

\begin{definition}\label{def:pastwitset}
Given a finite compass structure $\cG=(\bbP_\bbO,\cL)$ and a row $y$, a \emph{past witness set for $y$} is any \emph{minimal} set $\pastwit(y) \subseteq \bbP_\bbO$ such that for every request $\psi \in \req_{\bar{L}}(\obs(\cL(y-1,y))$ there exists a witness $(x_\psi,y_\psi)$ such that $\psi \in \obs(\cL(x_\psi,y_\psi))$ and $y_\psi < y - 1$.
\end{definition}

\noindent Again, by the minimality of $\pastwit(y)$ we have that there is at most one distinct point  
for every $\bar{L}$-formula in $\cL(y-1,y)$ and thus $\len{\pastwit(y)}\leq \len{\req_{\bar{L}}(y-1,y)}\leq \len{\closure(\varphi)}\leq 2\cdot \len{\varphi}$.

\medskip

We concentrate our attention on infinite structures that are unbounded both on the future and on the past (i.e., based on the set of integers $\bbZ$). The case when the structure is unbounded only in one direction (e.g., the naturals $\bbN$ or the set of negative integers $\bbZ^-$) can be tackled in a similar way by appropriately adapting the following notions and theorems.

\begin{definition}\label{def:compassgenerator}
Given an $\ABBL$ formula $\varphi$ and a finite, partially fulfilling compass structure $\cG=(\bbP_\bbO,\cL)$
of size  $N$, we say that $\cG$ is a \emph{compass generator for $\varphi$} if there exists four rows $y_\varphi$, $y_0$, $y_1$, and $y_2$ which satisfy the following properties:
\begin{compactenum}[\bf G1]
	\item $y_0 < y_1 < y_2$ and $y_0 \leq y_\varphi$,
	\item $\varphi \in \cL(y_ \varphi-1,y_ \varphi)$ or $\DBbar\varphi\in \cL(y_ \varphi-1,y_ \varphi)$,
	\item $\shading(y_1) \subseteq \shading(y_0)$ and $\cL(y_0-1, y_0) = \cL(y_1-1, y_1)$,
	\item there exists a past witness set  $\pastwit(y_1)$ such that $y_0 \leq \min(\pi_{y_1}(\pastwit(y_1)))$,
	\item $\shading(N-1) \subseteq \shading(y_2)$ and $\cL(y_2-1, y_2) = \cL(N-2, N-1)$,
	\item there exists a future witness set  $\futwit(y_2)$ for $y_2$.
\end{compactenum}
\end{definition}

% Figure~\ref{????} graphically describes the conditions under which a finite, 
% partially fulfilling compass structure is also a compass generator. 
The next theorem shows that the information contained in a compass generator for $\varphi$ is sufficient to build an infinite fulfilling compass structure featuring $\varphi$.

\begin{theorem}\label{thm:generateinfinite}
An $\ABBL$ formula $\varphi$ is satisfiable over the integers $\bbZ$ if and only if there exists 
a compass generator $\cG=(\bbP_\bbO,\cL)$ for $\varphi$.
\end{theorem}

\begin{proof}
\textbf{($\Rightarrow$)}
Let $\varphi$ an $\ABBL$ formula that is satisfiable over an infinite fulfilling compass structure
$\cG=(\bbP_\bbZ,\cL)$. Since $\cG$ features $\varphi$ we have that there exists a point $(x,y)$
with $\varphi \in \cL(x,y)$ and thus the row $y_\varphi = x + 1$ respects condition \textbf{G2}. 

Now, let $\Inf(\cG)$ be the set of shadings that occurs infinitely often in $\cG$. 
We define $y_1$ as the greatest row such that for every $y' \leq y_1$, $\shading(y') \in \Inf(\cG)$, and $y_2$ as the smallest row such that for every $y' \geq y_2$, $\shading(y') \in \Inf(\cG)$. Clearly, since $\cG$ is unbounded in the past, we can find two rows $y_{min}$ and $y_0$ such that $y_{min} < y_0$, and a corresponding portion of the grid $\bbP_{y_{min}} = \{(x,y) : x \geq y_{min}\}$ such that
 \begin{inparaenum}[\it (i)]
	\item $y_0 \leq y_\varphi$,
	\item $y_0 < y_1$,
	\item $\shading(y_1) \subseteq \shading(y_0)$ in $\bbP_{y_{min}}$,
	\item $\cL(y_0-1, y_0) = \cL(y_1-1, y_1)$, and
	\item there exists a past witness set  $\pastwit(y_1)$ for $y_1$ such that $y_0 \leq \min(\pi_{y_1}(\pastwit(y_1)))$ in $\bbP_{y_{min}}$.
\end{inparaenum}
Hence, conditions \textbf{G3} and \textbf{G4} are respected. 

Symmetrically, since $\cG$ is unbounded in the future, we can find a row $y_{max} > y_2$ and a corresponding portion of the grid $\bbP_{y_{min}}^{y_{max}} = \{(x,y) : x \geq y_{min} \land y \leq y_{max}\}$ such that
\begin{asparaenum}
	\item $\shading(y_{max}) \subseteq \shading(y_2)$,
	\item $\cL(y_2-1, y_2) = \cL(y_{max}-1, y_{max})$, and
	\item there exists a future witness set  $\futwit(y_2)$ for $y_2$ in $\bbP_{y_{min}}^{y_{max}}$.
\end{asparaenum}
This shows that conditions \textbf{G5} and \textbf{G6} are respected as well. Since $y_0 \leq y_\varphi$ and $y_0 < y_1 < y_2$ condition \textbf{G1} is also respected. Since the restriction of $\cG$ to the finite grid $\bbP_{y_{min}}^{y_{max}}$ is a partially fulfilling compass structure, we have found the required compass generator for $\varphi$.

\medskip

\textbf{($\Leftarrow$)}
Let $\cG=(\bbP_\bbO,\cL)$ be a compass generator of size $N$ for $\varphi$ and let $y_0<y_1<y_2$ and $y_\varphi$ 
be the four rows that satisfy properties \textbf{G1}--\textbf{G6} of Definition~\ref{def:compassgenerator}.
We will define an infinite sequence of partially fulfilling compass structures $\cG_0 \subset \cG_1 \subset \cG_2 \subset \ldots$ such that the infinite union $\cG^\omega = \bigcup_{i=0}^{+\infty} \cG_i$ is an infinite fulfilling compass structure that features $\varphi$. We start from the initial compass structure $\cG_0=(\bbP^0,\cL^0)$ where $\bbP^0=\{(x,y)\in\bbP_\bbO\ :\  x \geq y_0 - 1 \land y_0\leq y < N \}$ and $\cL^0(x,y)=\cL(x,y)$ for every point $(x,y)\in \bbP^0_\bbO$, and we will show how to iteratively build the infinite sequence of compass structures. For every step $i$ of the procedure, let $\cG_i=(\bbP^i,\cL^i)$ be the current structure, and let
$y^i_{min}$ and $y^i_{max}$ be the minimum and maximum vertical coordinate in $\bbP^i$, respectively. We guarantee that the following invariant is respected:

\begin{itemize}[\bf (INV)]
\item	$\shading_{\cG_i}(y^i_{max}) \subseteq \shading_{\cG}(y_2)$, \\
$\shading_{\cG_i}(y^i_{min}+ y_1 - y_0) \subseteq \shading_{\cG}(y_0)$, \\
$\cL^i(y^i_{max}-1, y^i_{max}) = \cL(y_2-1,y_2)$, and $\cL^i(y^i_{min}-1, y^i_{min}) = \cL(y_0-1,y_0)$.
\end{itemize}

\noindent The invariant trivially holds for $\cG_0$. Now, suppose that $\cG_i$ respects \textbf{(INV)} and let $k_{past} = y_1 - y_0$ and $k_{future} = N - y_2$. Figure~\ref{fig:compassgenerator} depicts how $\cG_{i+1}=(\bbP^{i+1},\cL^{i+1})$ can be built from $\cG_i$. Formally, the procedure is defined as follows.

\begin{figure}[t]
\centering
	\input{compass-generator-figure}
\caption{A compass generator (left) and a portion of the generated infinite compass structure (right).}
\label{fig:compassgenerator}
\end{figure}
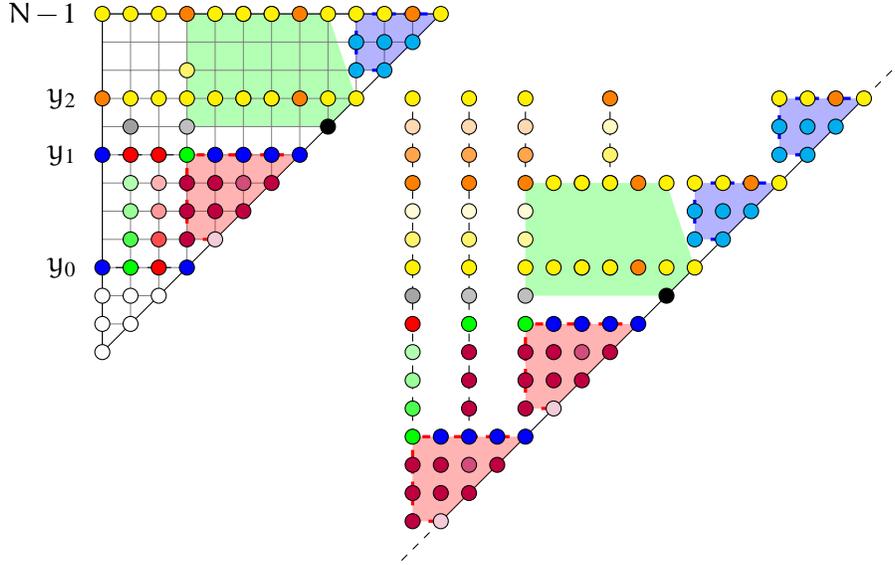

\begin{enumerate}[\quad a)]
	\item $y^{i+1}_{min} = y^{i}_{min} - k_{past}$, $y^{i+1}_{max} = y^{i}_{max} + k_{future}$, and $\bbP^{i+1} = \{(x,y)\in\bbP_\bbZ\ :\  x \geq y^{i+1}_{min} - 1 \land y^{i+1}_{min} \leq y < y^{i+1}_{max} \}$. 
	
	\item for every point $p \in \bbP^{i+1} \cap \bbP^{i}$, let $\cL^{i+1}(p)=\cL^i(p)$.
	
	\item for every point $(x,y) \in \bbP^{i+1} \setminus \bbP^{i}$ such that $y \leq y^i_{min}$, let $\cL^{i+1}(x,y)=\cL^i(x+k_{past},y+k_{past})$ (red area in Fig.~\ref{fig:compassgenerator}). 
	
	\item for every point $(x,y) \in \bbP^{i+1} \setminus \bbP^{i}$ such that $x \geq y^i_{max}$, let $\cL^{i+1}(x,y)=\cL^i(x-k_{future},y-k_{future})$ (blue area in Fig.~\ref{fig:compassgenerator}).
	
	\item By construction, for every point $(x,y^i_{min})$ with $x < y^i_{min}-1$ we have that $\cL^{i+1}(x,y^i_{min})=\cL^i(x+k_{past},y^i_{min}+k_{past})$. Since $\cG_i$ respects the invariant, $\cL^i(x+k_{past},y^i_{min}+k_{past}) = \cL^{i+1}(x,y^i_{min}) \in\shading_\cG(y_0)$. Let $(\bar{x},y_0)$ be a point on the row $y_0$ with the same labelling of $\cL^{i+1}(x,y^i_{min})$: we define the labelling of all points $(x, y^i_{min} + j)$, with $1\leq j\leq k_{past}$, as $\cL^{i+1}(x,y^i_{min}+j)=\cL(\bar{x},y_0+j)$. Now, since $\cL(\bar{x},y_0+k_{past}) \in \shading_\cG(y_1)$ ($y_1=y_0+k_{past}$) and $\shading_\cG(y_1)\subseteq\shading_\cG(y_0)$ (\textbf{G3}), we can find a point $(\hat{x},y_0)$ on the row $y_0$ with the same labelling of $\cL^{i+1}(x,y^i_{min}+k_{past})$ and define the labelling of every point $(x,y^i_{min}+k_{past}\cdot j)$ for every $1< j\leq i+1$. At the end of this procedure we have labelled all points $(x,y)$ such that $y \leq y_1$. 
	
	\item For every point $(x,y_1)$, by construction, we have that $\cL^{i+1}(x, y_1)\in \shading_\cG(y_1)$. Let $(\bar{x},y_1)$ be a point such that $\cL^{i+1}(x, y_1)=\cL(\bar{x}, y_1)$. As in the previous case, we define the labelling of all points $(x, y)$, with $y_1< y \leq y_2$ as $\cL^{i+1}(x,y)=\cL(\bar{x},y)$. At the end of this step we labelled all points $(x,y)$ such that $y \leq y_2$. 
	
	\item Now, by construction, for every point $(x,y_2)$ we have that $\cL^{i+1}(x,y_2)\in \shading_\cG(y_2)$. By condition \textbf{G6} of Definition~\ref{def:compassgenerator}, there exists a point $\bar{x} \in \futwit(y_2)$ such that $\cL^{i+1}(x,y_2)=\cL(\bar{x},y_2)$. We define $\cL^{i+1}(x,y_2 + j)=\cL(\bar{x},y_2 + j )$ for every $1 \leq j \leq k_{future}$. Since $y_2+k_{future}=N-1$ and $\shading(N-1)\subseteq \shading(y_2)$ we have that $\cL^{i+1}(x,N-1) \in \shading(y_2)$ (\textbf{G5}) and thus we can repeat this procedure iteratively until we have labelled all points $(x,y)$ such that $y \leq y^{i+1}_{max}$ and $x < y^i_{min}-1$.

	\item To conclude the procedure, we must define the labelling of points $(x,y)$ such that $x \geq y^i_{min}-1$ and $y \geq y^{i}_{max}$. Note that for every point $(x,y^i_{max})$ with $x\geq y^i_{min}-1$ we have, by the invariant, that $\shading_{\cG^i}(y^i_{max})\subseteq \shading_\cG$ $(y_2)$. Then there exists a point $\bar{x} \in \futwit(y_2)$ such that $\cL^{i+1}(x,y^i_{max})=\cL(\bar{x},y_2)$. We define $\cL^{i+1}(x,y^i_{max} + j)=\cL(\bar{x},y_2 + j )$ for every $1 \leq j \leq k_{future}$. 
\end{enumerate}

It is easy to see that $\cG_i$ is a partially fulfilling compass structure that respects the invariant.
Moreover, suppose that for some point $p = (x,y) \in \bbP^{i}$ and relation $R\in \{ A, B, \bar{B}, \bar{L} \}$ there exists  $\alpha\in \req_R(p)$ that is not fulfilled in $\cG_i$. We show that $\cG_{i+1}$ fulfills the $R$-request $\alpha$ for $p$.
\begin{compactitem}
	\item If $R=A$, since $\cG^i$ is partial fulfilling and it is finite we have that the point $p'=(y,y^i_{max})$ is such that $\alpha \in \req_{\bar{B}}(\cL(p'))$. By step h) of the procedure, and by the definition of future witness set, $\cG_{i+1}$ contains a point $p'' = (y,y^i_{max}+j)$ such that $\alpha \in \cL^{i+1}(p'')$.

	\item If $R=B$, by Definition \ref{def:partialfulfilling} we have all the $B$-requests in a partial fulfilling compass structure are fulfilled and thus this case connot be given.
	
\item If $R=\bar{B}$ the case is analogous to the case of $R=A$.

\item If $R=\bar{L}$, since  $\cG^i$ is partial fulfilling and it is finite we have that $\alpha \in \req_{\bar{L}}(\cL(y^i_{min}-1,y^i_{min}))$. By point c) of the construction we have that $\cL^i(y^i_{min}-1,y^i_{min})=\cL(y_0-1,y_0)=\cL(y_1-1,y_1)$. Hence, by condition \textbf{G4} of Definition~\ref{def:compassgenerator} and by the definition of past witness set, there exists a point $(\bar{x},\bar{y})$ with $y_0\leq \bar{x}<\bar{y}\leq y_1$ such that $\alpha \in \cL(\bar{x},\bar{y})$. By construction we have that $\cL(\bar{x},\bar{y})=\cL^{i+1}(\bar{x}-(i+1)\cdot k_{past},\bar{y}-(i+1)\cdot k_{past})$ and thus and thus the $\bar{L}$-request $\alpha$ for the point $p$ is fulfilled at step $i+1$ by the point $(\bar{x}- (i+1)\cdot k_{past},\bar{y}-(i+1)\cdot k_{past})$.
\end{compactitem}

Hence, we can conclude that the infinite compass structure $\cG^\omega$ is fulfilling. By condition \textbf{G2} of Definition~\ref{def:compassgenerator} we have that $\cG^\omega$ features $\varphi$ and thus that $\varphi$ is satisfiable over the integers.
\end{proof}

Theorem~\ref{thm:generateinfinite} shows that satisfiability of a formula over infinite models can be reduced to the existence of a finite compass generator for it. However, it does not give any bound on the size of it. In the following we will show how the techniques exploited in Section~\ref{subsec:finitecase} for finite models can be adapted to obtain a doubly exponential bound on the size of compass generators.

\begin{definition}\label{def:global-compatible-rows}
Given a compass generator $\cG=(\bbP_\bbO,\cL)$,  we say that two rows $y<y'$ are \emph{globally compatible} if and only if the following properties holds:
\begin{compactenum}
	\item  $\cL(y-1, y) = \cL(y' - 1, y')$ and $\shading_\cG(y) = \shading_\cG(y')$,
	\item for every $\bar{y} \in\{y_\varphi,y_0,y_1,y_2\}$ it is not the case that $y \leq \bar{y} \leq y'$,
	\item there exists a past witness set $\pastwit(y_1)$ such that for every point $(\bar{x},\bar{y})\in \pastwit(y_1)$ it is not the case that $y \leq \bar{y} \leq y'$;
	\item there exists a future witness set $\futwit(y_2)$ such that for every point $\bar{x}\in \futwit(y_2)$ and every $\bar{B}$-request $\alpha \in \req_{\bar{B}}(\cL(\bar{x},y_2)$ there is a point $(\bar{x},\bar{y})$ such that $y_2 < \bar{y}$, $\alpha \in \obs(\cL(\bar{x},y_2))$ and it is not the case that $y \leq \bar{y} \leq y'$;
	\item there exists a witness set $\witset(y')$ for $y'$ and an \emph{injective mapping function} $w: \pi_{y'}(\witset(y')\cup\pastwit(y_1)\cup \futwit(y_2)) \mapsto \{x : x < y \}$, such that $\cL(x,y')=\cL(w(x),y)$, for every \linebreak $x\in\pi_{y'}(\witset(y')\cup\pastwit(y_1)\cup \futwit(y_2))$, and $w(x)=x$, for every 
	$x \in\pi_{y'}(\pastwit(y_1)$.
\end{compactenum}
\end{definition}

Clearly, two globally compatible rows are compatible. The additional conditions of the definition guarantees that the contraction procedure do not remove ``meaningful'' parts of the compass generator, like the rows $y_\varphi$, $y_0$, $y_1$, and $y_2$ (condition $2$) or future and past witnesses (conditions $3$ and $4$).

\begin{lemma}\label{lemma:contraction-infinite}
Let $\cG$ be a compass generator for $\varphi$ of size $N$. If there exist two global-compatible rows $0<y<y'<N$ in $\cG$, then there exists a compass generator $\cG'$ of size $N' = N - y + y'$ that features $\varphi$.
\end{lemma}
\begin{proof}
We can define a function $f:\{0,...,y\}\then\{0,...,y'\}$ and contract $\cG$ to a smaller compass structure $\cG'$  in the very same way of Lemma \ref{lemma:contraction_finite}. It can be easily proved that the obtained $\cG'$ is a partial fulfilling compass structure. Let $k = y'-y$ and let $y'_\varphi=y_\varphi$ if $y_\varphi<y$, $y'_\varphi=y_\varphi-k$ otherwise. To prove that $\cG'$ is a compass generator, let us consider the following four cases.

\begin{compactitem}
\item[-] If $y'<y_0$, then we have that $y'_i=y_i-k$ for $i \in \{0,1,2,\varphi\}$ satisfy conditions  \textbf{G1-G6} in $\cG'$.

\item[-] If $y_0<y<y'<y_1$, then for every point $(\bar{x},\bar{y}) \in \pastwit(y_1)$ we have that either $f(\bar{x},\bar{y})=(\bar{x},\bar{y})$ (when $\bar{y} < y$) or $f(\bar{x},\bar{y}-k)=(w(\bar{x}),\bar{y}-k)=(\bar{x},\bar{y})$ (when $\bar{y} > y'$), and thus $\pastwit(y_1)$ is a past witness set for $\cG'$ as well. 
From this we can conclude that $y'_\varphi,y_0,y_1-k,$ and $ y_2-k $ satisfy conditions  \textbf{G1-G6} in $\cG'$.

 \item[-] If $y_0<y_1<y<y'<y_2$, then it is easy to prove that $y'_\varphi,y_0,y_1$ and $y_2-k$  satisfy \textbf{G1-G6} in $\cG'$.
 
 \item[-] If $y_0<y_1<y_2<y<y'$, then it is easy to observe that $y'_\varphi,y_0,y_1$ and $y_2$  satisfy \textbf{G1-G6} in $\cG'$.
\end{compactitem}

\noindent Hence, in all possible cases $\cG'$ is a compass generator for $\varphi$.
\end{proof}

\begin{theorem}\label{th:contraction-infinite}
An $\ABBL$-formula $\varphi$ is satisfied by some infinite interval structure iff it is 
featured by some compass generator of length $N\le (2|\varphi|+1)^{2^{8\len{\varphi}}}\cdot 2^{16\len{\varphi}^2 + 8\len{\varphi}}$ (i.e., double 
exponential in $\len{\varphi}$). 
\end{theorem}

\begin{proof}

Suppose that $\varphi$ is satisfied by a infinite interval structure $\cS$. 
By Theorem \ref{thm:generateinfinite}, there is a compass generator $\cG$ that features $\varphi$.
By Lemma~\ref{lemma:contraction-infinite}, we can assume without loss of generality that all rows of $\cG$ are pairwise global-incompatible. Let $c_y$ the characteristic function defined in the proof of Theorem~\ref{th:contraction_finite}. Now, let $x_1 < \ldots < x_k$ be the ordered sequence of the points in $\pastwit(y_1)$. We associate to every row $y$ a finite word $W_y$ of length $|W_y|\leq k \leq 2 \cdot |\varphi|$
on the alphabet $\cA_{\varphi}$ ($|\cA_{\varphi}|= 2^{8\len{\varphi}}$) such that for every $x_i \in \pastwit(y_1)$, $W(i)=\cL(x_i,y)$. It is easy to prove that two rows $y<y'$ in $\bbO$ with $c_y(F)=c_{y'}$, $W_{y}=W_{y'}$ and such that $\cL(y'-1,y') = \cL(y-1,y)$ are global-compatible.  

Since the number of possible characteristic functions is bounded by $(2|\varphi|+1)^{2^{8\len{\varphi}}}$, and the number of possible words is bounded by  $(2^{8\len{\varphi}})^{2 \cdot\len{\varphi}} = 2^{16\len{\varphi}^2}$,
$\cG$ cannot have more than $(2|\varphi|+1)^{2^{8\len{\varphi}}}\cdot 2^{16\len{\varphi}^2 + 8\len{\varphi}}$ rows, and thus $N$ is at most doubly exponential in $\len{\varphi}$.
\end{proof}

%% file: compass-generator-figure.tex
\begin{tikzpicture}[scale=0.75]
\fill[green, opacity=0.3] (1,1) -- (-1.5,1) -- (-1.5,3) -- (1, 3) -- (1.5,1.5) --cycle;
\draw[step=0.5cm,gray,very thin] (-3,-3) grid (3,3);
\fill[color=white] (-3.1,-3.1) -- (3.1,3.1) -- (3.1,-3.1);
\draw (-3,-3) -- (3,3);
\draw (-3,-3) -- (-3,3);
\draw (-3,3) -- (3,3);

\draw[dashed] (-1.5,-1.5) -- (-3, -1.5);
\draw[dashed] (-3,0.5) -- (-1.5,0.5);
\fill[red, opacity=0.3] (-1.5,-1) -- (-1, -1) -- (0.5,0.5) -- (-1.5, 0.5) --cycle;
\draw[dashed,red, very thick] (-1.5,-1) -- (-1, -1);
\draw[dashed,red, very thick] (-1.5,-1) -- (-1.5, 0.5);
\draw[dashed,red, very thick] (0.5,0.5) -- (-1.5, 0.5);

\fill[blue, opacity=0.3] (2,2) -- (1.5, 2) -- (1.5,3) -- (3, 3) --cycle;
\draw[dashed,blue, very thick] (1.5,2) -- (2, 2);
\draw[dashed,blue, very thick] (1.5,2) -- (1.5, 3);
\draw[dashed,blue, very thick] (1.5,3) -- (3, 3);

%\draw[dashed] (-0.5,1) -- (-0.5, 2);
%\draw[dashed] (-0.5,1) -- (1, 1);
%\draw[dashed] (1,2) -- (1, 1);
%\draw[dashed] (-3,-0.5) -- (-0.5, -0.5);
\node[shape=circle,draw=black,inner sep=2pt,fill=white, label={[label distance=0.1cm]above left:$$}](A) at (-3,-3) {};
\node[shape=circle,draw=black,inner sep=2pt,fill=white, label={[label distance=0.1cm]above left:$$}](A) at (-2.5,-2.5) {};
\node[shape=circle,draw=black,inner sep=2pt,fill=white, label={[label distance=0.1cm]above left:$$}](A) at (-3,-2.5) {};
\node[shape=circle,draw=black,inner sep=2pt,fill=white, label={[label distance=0.1cm]above left:$$}](A) at (-2.5,-2) {};
\node[shape=circle,draw=black,inner sep=2pt,fill=white, label={[label distance=0.1cm]above left:$$}](A) at (-2,-2) {};
\node[shape=circle,draw=black,inner sep=2pt,fill=white, label={[label distance=0.1cm]above left:$$}](A) at (-3,-2) {};

\node[shape=circle,draw=black,inner sep=2pt,fill=blue, label={[label distance=0.1cm]above left:$$}](A) at (-1.5,-1.5) {};

\node[shape=circle,draw=black,inner sep=2pt,fill=red, label={[label distance=0.1cm]above left:$$}](A) at (-2,-1.5) {};
\node[shape=circle,draw=black,inner sep=2pt,fill=red!70, label={[label distance=0.1cm]above left:$$}](A) at (-2,-1) {};
\node[shape=circle,draw=black,inner sep=2pt,fill=red!40, label={[label distance=0.1cm]above left:$$}](A) at (-2,-0.5) {};
\node[shape=circle,draw=black,inner sep=2pt,fill=red!30, label={[label distance=0.1cm]above left:$$}](A) at (-2,0) {};

\node[shape=circle,draw=black,inner sep=2pt,fill=green, label={[label distance=0.1cm]above left:$$}](A) at (-2.5,-1.5) {};
\node[shape=circle,draw=black,inner sep=2pt,fill=green!70, label={[label distance=0.1cm]above left:$$}](A) at (-2.5,-1) {};
\node[shape=circle,draw=black,inner sep=2pt,fill=green!40, label={[label distance=0.1cm]above left:$$}](A) at (-2.5,-0.5) {};
\node[shape=circle,draw=black,inner sep=2pt,fill=green!30, label={[label distance=0.1cm]above left:$$}](A) at (-2.5,0) {};

\node[shape=circle,draw=black,inner sep=2pt,fill=blue, label={[label distance=0.1cm]left:$y_0$}](A) at (-3,-1.5) {};

\node[shape=circle,draw=black,inner sep=2pt,fill=purple!20, label={[label distance=0.1cm]left:$$}](P1) at (-1,-1) {};
\node[shape=circle,draw=black,inner sep=2pt,fill=purple, label={[label distance=0.1cm]left:$$}](A) at (-1.5,-1) {};
\node[shape=circle,draw=black,inner sep=2pt,fill=purple, label={[label distance=0.1cm]left:$$}](A) at (-1,-0.5) {};
\node[shape=circle,draw=black,inner sep=2pt,fill=purple, label={[label distance=0.1cm]left:$$}](A) at (-0.5,-0.5) {};
\node[shape=circle,draw=black,inner sep=2pt,fill=purple, label={[label distance=0.1cm]left:$$}](A) at (-1.5,-0.5) {};
\node[shape=circle,draw=black,inner sep=2pt,fill=purple, label={[label distance=0.1cm]left:$$}](A) at (-1,-0) {};
\node[shape=circle,draw=black,inner sep=2pt,fill=purple!70, label={[label distance=0.1cm]left:$$}](P2) at (-0.5,-0) {};
\node[shape=circle,draw=black,inner sep=2pt,fill=purple, label={[label distance=0.1cm]left:$$}](A) at (-1.5,-0) {};
\node[shape=circle,draw=black,inner sep=2pt,fill=purple, label={[label distance=0.1cm]left:$$}](A) at (0,0) {};

%\node(P1F) [below right of=P1, node distance=1.5cm] {$\psi_5$};
%\node(P2F)[below right of=P2, node distance=1.5cm] {$\psi_4$};

%\path[->, very thick] (P1F) edge[bend left](P1);
%\path[->, very thick] (P2F) edge(P2);

\node[shape=circle,draw=black,inner sep=2pt,fill=blue, label={[label distance=0.1cm]left:$$}](A) at (0.5,0.5) {};
\node[shape=circle,draw=black,inner sep=2pt,fill=blue, label={[label distance=0.1cm]left:$$}](A) at (0,0.5) {};
\node[shape=circle,draw=black,inner sep=2pt,fill=blue, label={[label distance=0.1cm]left:$$}](A) at (-0.5,0.5) {};
\node[shape=circle,draw=black,inner sep=2pt,fill=blue, label={[label distance=0.1cm]left:$$}](A) at (-1,0.5) {};
\node[shape=circle,draw=black,inner sep=2pt,fill=green, label={[label distance=0.1cm]left:$$}](A) at (-1.5,0.5) {};
\node[shape=circle,draw=black,inner sep=2pt,fill=red, label={[label distance=0.1cm]left:$$}](A) at (-2,0.5) {};
\node[shape=circle,draw=black,inner sep=2pt,fill=red, label={[label distance=0.1cm]left:$$}](A) at (-2.5,0.5) {};
\node[shape=circle,draw=black,inner sep=2pt,fill=blue, label={[label distance=0.1cm]left:$y_1$}](A) at (-3,0.5) {};

\node[shape=circle,draw=black,inner sep=2pt,fill=yellow, label={[label distance=0.1cm]left:$$}](A) at (1.5,1.5) {};

\node[shape=circle,draw=black,inner sep=2pt,fill=cyan, label={[label distance=0.1cm]left:$$}](A) at (2,2) {};
\node[shape=circle,draw=black,inner sep=2pt,fill=cyan, label={[label distance=0.1cm]left:$$}](A) at (1.5,2) {};
\node[shape=circle,draw=black,inner sep=2pt,fill=cyan, label={[label distance=0.1cm]left:$$}](A) at (2.5,2.5) {};
\node[shape=circle,draw=black,inner sep=2pt,fill=cyan, label={[label distance=0.1cm]left:$$}](A) at (2,2.5) {};
\node[shape=circle,draw=black,inner sep=2pt,fill=cyan, label={[label distance=0.1cm]left:$$}](A) at (1.5,2.5) {};
\node[shape=circle,draw=black,inner sep=2pt,fill=yellow, label={[label distance=0.1cm]left:$$}](A) at (3,3) {};
\node[shape=circle,draw=black,inner sep=2pt,fill=orange, label={[label distance=0.1cm]left:$$}](A) at (2.5,3) {};
\node[shape=circle,draw=black,inner sep=2pt,fill=yellow, label={[label distance=0.1cm]left:$$}](A) at (2,3) {};
\node[shape=circle,draw=black,inner sep=2pt,fill=yellow, label={[label distance=0.1cm]left:$$}](A) at (1.5,3) {};

\node[shape=circle,draw=black,inner sep=2pt,fill=yellow, label={[label distance=0.1cm]left:$$}](A) at (1,3) {};
\node[shape=circle,draw=black,inner sep=2pt,fill=orange, label={[label distance=0.1cm]left:$$}](A) at (0.5,3) {};
\node[shape=circle,draw=black,inner sep=2pt,fill=yellow, label={[label distance=0.1cm]left:$$}](A) at (0,3) {};
\node[shape=circle,draw=black,inner sep=2pt,fill=orange, label={[label distance=0.1cm]left:$$}](A) at (-0.5,3) {};
\node[shape=circle,draw=black,inner sep=2pt,fill=yellow, label={[label distance=0.1cm]left:$$}](A) at (0,3) {};
\node[shape=circle,draw=black,inner sep=2pt,fill=yellow, label={[label distance=0.1cm]left:$$}](A) at (-0.5,3) {};
\node[shape=circle,draw=black,inner sep=2pt,fill=yellow, label={[label distance=0.1cm]left:$$}](A) at (-1,3) {};
\node[shape=circle,draw=black,inner sep=2pt,fill=orange, label={[label distance=0.1cm]left:$$}](A) at (-1.5,3) {};
\node[shape=circle,draw=black,inner sep=2pt,fill=yellow, label={[label distance=0.1cm]left:$$}](A) at (-2,3) {};
\node[shape=circle,draw=black,inner sep=2pt,fill=yellow, label={[label distance=0.1cm]left:$$}](A) at (-2.5,3) {};
\node[shape=circle,draw=black,inner sep=2pt,fill=yellow, label={[label distance=0.1cm]left:$N-1$}](A) at (-3,3) {};

\node[shape=circle,draw=black,inner sep=2pt,fill=yellow, label={[label distance=0.1cm]left:$$}](A) at (1,1.5) {};
\node[shape=circle,draw=black,inner sep=2pt,fill=orange, label={[label distance=0.1cm]left:$$}](A) at (0.5,1.5) {};
\node[shape=circle,draw=black,inner sep=2pt,fill=yellow, label={[label distance=0.1cm]left:$$}](A) at (0,1.5) {};
\node[shape=circle,draw=black,inner sep=2pt,fill=yellow, label={[label distance=0.1cm]left:$$}](A) at (-0.5,1.5) {};
\node[shape=circle,draw=black,inner sep=2pt,fill=yellow, label={[label distance=0.1cm]left:$$}](A) at (0,1.5) {};
\node[shape=circle,draw=black,inner sep=2pt,fill=yellow, label={[label distance=0.1cm]left:$$}](A) at (-0.5,1.5) {};
\node[shape=circle,draw=black,inner sep=2pt,fill=yellow, label={[label distance=0.1cm]left:$$}](A) at (-1,1.5) {};
\node[shape=circle,draw=black,inner sep=2pt,fill=yellow, label={[label distance=0.1cm]left:$$}](A) at (-1.5,1.5) {};
\node[shape=circle,draw=black,inner sep=2pt,fill=yellow, label={[label distance=0.1cm]left:$$}](A) at (-2,1.5) {};
\node[shape=circle,draw=black,inner sep=2pt,fill=yellow, label={[label distance=0.1cm]left:$$}](A) at (-2.5,1.5) {};
\node[shape=circle,draw=black,inner sep=2pt,fill=orange, label={[label distance=0.1cm]left:$y_2$}](A) at (-3,1.5) {};

%\node[shape=circle,draw=black,inner sep=2pt,fill=orange!70, label={[label distance=0.1cm]right:$\psi_1$}](A) at (-3,2) {};
%\node[shape=circle,draw=black,inner sep=2pt,fill=orange!30, label={[label distance=0.1cm]right:$\psi_2$}](A) at (-3,2.5) {};

\node[shape=circle,draw=black,inner sep=2pt,fill=yellow!70, label={[label distance=0.1cm]right:$$}](A) at (-1.5,2) {};
%\node[shape=circle,draw=black,inner sep=2pt,fill=yellow!30, label={[label distance=0.1cm]right:$\psi_3$}](A) at (-1.5,2.5) {};

\node[shape=circle,draw=black,inner sep=2pt,fill=gray!50, label={[label distance=0.1cm]right:$$}](A) at (-1.5,1) {};
\node[shape=circle,draw=black,inner sep=2pt,fill=black, label={[label distance=0.1cm]right:$$}](A) at (1,1) {};
\node[shape=circle,draw=black,inner sep=2pt,fill=gray!75, label={[label distance=0.1cm]right:$$}](A) at (-2.5,1) {};

\pgftransformshift{\pgfpoint{6cm}{-3cm}}

\fill[green, opacity=0.3] (1,1) -- (-1.5,1) -- (-1.5,3) -- (1, 3) -- (1.5,1.5) --cycle;

\draw[dashed] (-1.5,3) -- (-1.5,4.5);
\draw[dashed] (0,3) -- (0,4.5);

\draw (-3,-3) -- (4.5,4.5);
\draw[dashed] (-3,-3) -- (-3.7,-3.7);
\draw[dashed] (4.5,4.5)-- (5.2,5.2);

\node[shape=circle,draw=black,inner sep=2pt,fill=gray!50, label={[label distance=0.1cm]right:$$}](A) at (-1.5,1) {};
\node[shape=circle,draw=black,inner sep=2pt,fill=black, label={[label distance=0.1cm]right:$$}](A) at (1,1) {};
\node[shape=circle,draw=black,inner sep=2pt,fill=yellow!50, label={[label distance=0.1cm]right:$$}](A) at (-1.5,2) {};
\node[shape=circle,draw=black,inner sep=2pt,fill=yellow!20, label={[label distance=0.1cm]right:$$}](A) at (-1.5,2.5) {};

\node[shape=circle,draw=black,inner sep=2pt,fill=yellow, label={[label distance=0.1cm]left:$$}](A) at (1.5,1.5) {};
\node[shape=circle,draw=black,inner sep=2pt,fill=yellow, label={[label distance=0.1cm]left:$$}](A) at (1,1.5) {};
\node[shape=circle,draw=black,inner sep=2pt,fill=orange, label={[label distance=0.1cm]left:$$}](A) at (0.5,1.5) {};
\node[shape=circle,draw=black,inner sep=2pt,fill=yellow, label={[label distance=0.1cm]left:$$}](A) at (0,1.5) {};
\node[shape=circle,draw=black,inner sep=2pt,fill=yellow, label={[label distance=0.1cm]left:$$}](A) at (-0.5,1.5) {};
\node[shape=circle,draw=black,inner sep=2pt,fill=yellow, label={[label distance=0.1cm]left:$$}](A) at (0,1.5) {};
\node[shape=circle,draw=black,inner sep=2pt,fill=yellow, label={[label distance=0.1cm]left:$$}](A) at (-0.5,1.5) {};
\node[shape=circle,draw=black,inner sep=2pt,fill=yellow, label={[label distance=0.1cm]left:$$}](A) at (-1,1.5) {};
\node[shape=circle,draw=black,inner sep=2pt,fill=yellow, label={[label distance=0.1cm]left:$$}](A) at (-1.5,1.5) {};

\node[shape=circle,draw=black,inner sep=2pt,fill=yellow, label={[label distance=0.1cm]left:$$}](A) at (1,3) {};
\node[shape=circle,draw=black,inner sep=2pt,fill=orange, label={[label distance=0.1cm]left:$$}](A) at (0.5,3) {};
\node[shape=circle,draw=black,inner sep=2pt,fill=yellow, label={[label distance=0.1cm]left:$$}](A) at (0,3) {};
\node[shape=circle,draw=black,inner sep=2pt,fill=yellow, label={[label distance=0.1cm]left:$$}](A) at (-0.5,3) {};
\node[shape=circle,draw=black,inner sep=2pt,fill=yellow, label={[label distance=0.1cm]left:$$}](A) at (0,3) {};
\node[shape=circle,draw=black,inner sep=2pt,fill=yellow, label={[label distance=0.1cm]left:$$}](A) at (-0.5,3) {};
\node[shape=circle,draw=black,inner sep=2pt,fill=yellow, label={[label distance=0.1cm]left:$$}](A) at (-1,3) {};
\node[shape=circle,draw=black,inner sep=2pt,fill=orange, label={[label distance=0.1cm]left:$$}](A) at (-1.5,3) {};

\fill[red, opacity=0.3] (-1.5,-1) -- (-1, -1) -- (0.5,0.5) -- (-1.5, 0.5) --cycle;
\draw[dashed,red, very thick] (-1.5,-1) -- (-1, -1);
\draw[dashed,red, very thick] (-1.5,-1) -- (-1.5, 0.5);
\draw[dashed,red, very thick] (0.5,0.5) -- (-1.5, 0.5);
\node[shape=circle,draw=black,inner sep=2pt,fill=blue, label={[label distance=0.1cm]left:$$}](A) at (0.5,0.5) {};
\node[shape=circle,draw=black,inner sep=2pt,fill=blue, label={[label distance=0.1cm]left:$$}](A) at (0,0.5) {};
\node[shape=circle,draw=black,inner sep=2pt,fill=blue, label={[label distance=0.1cm]left:$$}](A) at (-0.5,0.5) {};
\node[shape=circle,draw=black,inner sep=2pt,fill=blue, label={[label distance=0.1cm]left:$$}](A) at (-1,0.5) {};
\node[shape=circle,draw=black,inner sep=2pt,fill=green, label={[label distance=0.1cm]left:$$}](A) at (-1.5,0.5) {};
\node[shape=circle,draw=black,inner sep=2pt,fill=purple!20, label={[label distance=0.1cm]left:$$}](A) at (-1,-1) {};
\node[shape=circle,draw=black,inner sep=2pt,fill=purple, label={[label distance=0.1cm]left:$$}](A) at (-1.5,-1) {};
\node[shape=circle,draw=black,inner sep=2pt,fill=purple, label={[label distance=0.1cm]left:$$}](A) at (-1,-0.5) {};
\node[shape=circle,draw=black,inner sep=2pt,fill=purple, label={[label distance=0.1cm]left:$$}](A) at (-0.5,-0.5) {};
\node[shape=circle,draw=black,inner sep=2pt,fill=purple, label={[label distance=0.1cm]left:$$}](A) at (-1.5,-0.5) {};
\node[shape=circle,draw=black,inner sep=2pt,fill=purple, label={[label distance=0.1cm]left:$$}](A) at (-1,-0) {};
\node[shape=circle,draw=black,inner sep=2pt,fill=purple!70, label={[label distance=0.1cm]left:$$}](A) at (-0.5,-0) {};
\node[shape=circle,draw=black,inner sep=2pt,fill=purple, label={[label distance=0.1cm]left:$$}](A) at (-1.5,-0) {};
\node[shape=circle,draw=black,inner sep=2pt,fill=purple, label={[label distance=0.1cm]left:$$}](A) at (0,0) {};

\fill[blue, opacity=0.3] (2,2) -- (1.5, 2) -- (1.5,3) -- (3, 3) --cycle;
\draw[dashed,blue, very thick] (1.5,2) -- (2, 2);
\draw[dashed,blue, very thick] (1.5,2) -- (1.5, 3);
\draw[dashed,blue, very thick] (1.5,3) -- (3, 3);

\node[shape=circle,draw=black,inner sep=2pt,fill=cyan, label={[label distance=0.1cm]left:$$}](A) at (2,2) {};
\node[shape=circle,draw=black,inner sep=2pt,fill=cyan, label={[label distance=0.1cm]left:$$}](A) at (1.5,2) {};
\node[shape=circle,draw=black,inner sep=2pt,fill=cyan, label={[label distance=0.1cm]left:$$}](A) at (2.5,2.5) {};
\node[shape=circle,draw=black,inner sep=2pt,fill=cyan, label={[label distance=0.1cm]left:$$}](A) at (2,2.5) {};
\node[shape=circle,draw=black,inner sep=2pt,fill=cyan, label={[label distance=0.1cm]left:$$}](A) at (1.5,2.5) {};
\node[shape=circle,draw=black,inner sep=2pt,fill=yellow, label={[label distance=0.1cm]left:$$}](A) at (3,3) {};
\node[shape=circle,draw=black,inner sep=2pt,fill=orange, label={[label distance=0.1cm]left:$$}](A) at (2.5,3) {};
\node[shape=circle,draw=black,inner sep=2pt,fill=yellow, label={[label distance=0.1cm]left:$$}](A) at (2,3) {};
\node[shape=circle,draw=black,inner sep=2pt,fill=yellow, label={[label distance=0.1cm]left:$$}](A) at (1.5,3) {};

\pgftransformshift{\pgfpoint{1.5cm}{1.5cm}}

\fill[blue, opacity=0.3] (2,2) -- (1.5, 2) -- (1.5,3) -- (3, 3) --cycle;
\draw[dashed,blue, very thick] (1.5,2) -- (2, 2);
\draw[dashed,blue, very thick] (1.5,2) -- (1.5, 3);
\draw[dashed,blue, very thick] (1.5,3) -- (3, 3);

\node[shape=circle,draw=black,inner sep=2pt,fill=cyan, label={[label distance=0.1cm]left:$$}](A) at (2,2) {};
\node[shape=circle,draw=black,inner sep=2pt,fill=cyan, label={[label distance=0.1cm]left:$$}](A) at (1.5,2) {};
\node[shape=circle,draw=black,inner sep=2pt,fill=cyan, label={[label distance=0.1cm]left:$$}](A) at (2.5,2.5) {};
\node[shape=circle,draw=black,inner sep=2pt,fill=cyan, label={[label distance=0.1cm]left:$$}](A) at (2,2.5) {};
\node[shape=circle,draw=black,inner sep=2pt,fill=cyan, label={[label distance=0.1cm]left:$$}](A) at (1.5,2.5) {};
\node[shape=circle,draw=black,inner sep=2pt,fill=yellow, label={[label distance=0.1cm]left:$$}](A) at (3,3) {};
\node[shape=circle,draw=black,inner sep=2pt,fill=orange, label={[label distance=0.1cm]left:$$}](A) at (2.5,3) {};
\node[shape=circle,draw=black,inner sep=2pt,fill=yellow, label={[label distance=0.1cm]left:$$}](A) at (2,3) {};
\node[shape=circle,draw=black,inner sep=2pt,fill=yellow, label={[label distance=0.1cm]left:$$}](A) at (1.5,3) {};

\pgftransformshift{\pgfpoint{-3.5cm}{-3.5cm}}
\draw[dashed] (-0.5,0.5) -- (-0.5,6.5);
\draw[dashed] (-1.5,0.5) -- (-1.5,6.5);

\fill[red, opacity=0.3] (-1.5,-1) -- (-1, -1) -- (0.5,0.5) -- (-1.5, 0.5) --cycle;
\draw[dashed,red, very thick] (-1.5,-1) -- (-1, -1);
\draw[dashed,red, very thick] (-1.5,-1) -- (-1.5, 0.5);
\draw[dashed,red, very thick] (0.5,0.5) -- (-1.5, 0.5);
\node[shape=circle,draw=black,inner sep=2pt,fill=blue, label={[label distance=0.1cm]left:$$}](A) at (0.5,0.5) {};
\node[shape=circle,draw=black,inner sep=2pt,fill=blue, label={[label distance=0.1cm]left:$$}](A) at (0,0.5) {};
\node[shape=circle,draw=black,inner sep=2pt,fill=blue, label={[label distance=0.1cm]left:$$}](A) at (-0.5,0.5) {};
\node[shape=circle,draw=black,inner sep=2pt,fill=blue, label={[label distance=0.1cm]left:$$}](A) at (-1,0.5) {};
\node[shape=circle,draw=black,inner sep=2pt,fill=green, label={[label distance=0.1cm]left:$$}](A) at (-1.5,0.5) {};
\node[shape=circle,draw=black,inner sep=2pt,fill=purple!20, label={[label distance=0.1cm]left:$$}](A) at (-1,-1) {};
\node[shape=circle,draw=black,inner sep=2pt,fill=purple, label={[label distance=0.1cm]left:$$}](A) at (-1.5,-1) {};
\node[shape=circle,draw=black,inner sep=2pt,fill=purple, label={[label distance=0.1cm]left:$$}](A) at (-1,-0.5) {};
\node[shape=circle,draw=black,inner sep=2pt,fill=purple, label={[label distance=0.1cm]left:$$}](A) at (-0.5,-0.5) {};
\node[shape=circle,draw=black,inner sep=2pt,fill=purple, label={[label distance=0.1cm]left:$$}](A) at (-1.5,-0.5) {};
\node[shape=circle,draw=black,inner sep=2pt,fill=purple, label={[label distance=0.1cm]left:$$}](A) at (-1,-0) {};
\node[shape=circle,draw=black,inner sep=2pt,fill=purple!70, label={[label distance=0.1cm]left:$$}](A) at (-0.5,-0) {};
\node[shape=circle,draw=black,inner sep=2pt,fill=purple, label={[label distance=0.1cm]left:$$}](A) at (-1.5,-0) {};
\node[shape=circle,draw=black,inner sep=2pt,fill=purple, label={[label distance=0.1cm]left:$$}](A) at (0,0) {};

\node[shape=circle,draw=black,inner sep=2pt,fill=green!70, label={[label distance=0.1cm]above left:$$}](A) at (-1.5,1) {};
\node[shape=circle,draw=black,inner sep=2pt,fill=green!40, label={[label distance=0.1cm]above left:$$}](A) at (-1.5,1.5) {};
\node[shape=circle,draw=black,inner sep=2pt,fill=green!30, label={[label distance=0.1cm]above left:$$}](A) at (-1.5,2) {};
\node[shape=circle,draw=black,inner sep=2pt,fill=red, label={[label distance=0.1cm]above left:$$}](A) at (-1.5,2.5) {};
\node[shape=circle,draw=black,inner sep=2pt,fill=gray!70, label={[label distance=0.1cm]above left:$$}](A) at (-1.5,3) {};
\node[shape=circle,draw=black,inner sep=2pt,fill=yellow, label={[label distance=0.1cm]above left:$$}](A) at (-1.5,3.5) {};
\node[shape=circle,draw=black,inner sep=2pt,fill=yellow!70, label={[label distance=0.1cm]above left:$$}](A) at (-1.5,4) {};
\node[shape=circle,draw=black,inner sep=2pt,fill=yellow!20, label={[label distance=0.1cm]above left:$$}](A) at (-1.5,4.5) {};
\node[shape=circle,draw=black,inner sep=2pt,fill=orange, label={[label distance=0.1cm]above left:$$}](A) at (-1.5,5) {};
\node[shape=circle,draw=black,inner sep=2pt,fill=orange!70, label={[label distance=0.1cm]above left:$$}](A) at (-1.5,5.5) {};
\node[shape=circle,draw=black,inner sep=2pt,fill=orange!30, label={[label distance=0.1cm]above left:$$}](A) at (-1.5,6) {};
\node[shape=circle,draw=black,inner sep=2pt,fill=yellow, label={[label distance=0.1cm]above left:$$}](A) at (-1.5,6.5) {};

\node[shape=circle,draw=black,inner sep=2pt,fill=purple, label={[label distance=0.1cm]above left:$$}](A) at (-0.5,1) {};
\node[shape=circle,draw=black,inner sep=2pt,fill=purple, label={[label distance=0.1cm]above left:$$}](A) at (-0.5,1.5) {};
\node[shape=circle,draw=black,inner sep=2pt,fill=purple, label={[label distance=0.1cm]above left:$$}](A) at (-0.5,2) {};
\node[shape=circle,draw=black,inner sep=2pt,fill=green, label={[label distance=0.1cm]above left:$$}](A) at (-0.5,2.5) {};
\node[shape=circle,draw=black,inner sep=2pt,fill=gray!50, label={[label distance=0.1cm]above left:$$}](A) at (-0.5,3) {};
\node[shape=circle,draw=black,inner sep=2pt,fill=yellow, label={[label distance=0.1cm]above left:$$}](A) at (-0.5,3.5) {};
\node[shape=circle,draw=black,inner sep=2pt,fill=yellow!70, label={[label distance=0.1cm]above left:$$}](A) at (-0.5,4) {};
\node[shape=circle,draw=black,inner sep=2pt,fill=yellow!20, label={[label distance=0.1cm]above left:$$}](A) at (-0.5,4.5) {};
\node[shape=circle,draw=black,inner sep=2pt,fill=orange, label={[label distance=0.1cm]above left:$$}](A) at (-0.5,5) {};
\node[shape=circle,draw=black,inner sep=2pt,fill=orange!70, label={[label distance=0.1cm]above left:$$}](A) at (-0.5,5.5) {};
\node[shape=circle,draw=black,inner sep=2pt,fill=orange!30, label={[label distance=0.1cm]above left:$$}](A) at (-0.5,6) {};
\node[shape=circle,draw=black,inner sep=2pt,fill=yellow, label={[label distance=0.1cm]above left:$$}](A) at (-0.5,6.5) {};

\node[shape=circle,draw=black,inner sep=2pt,fill=yellow!70, label={[label distance=0.1cm]above left:$$}](A) at (2,5.5) {};
\node[shape=circle,draw=black,inner sep=2pt,fill=yellow!30, label={[label distance=0.1cm]above left:$$}](A) at (2,6) {};
\node[shape=circle,draw=black,inner sep=2pt,fill=orange, label={[label distance=0.1cm]above left:$$}](A) at (2,6.5) {};

\node[shape=circle,draw=black,inner sep=2pt,fill=orange!70, label={[label distance=0.1cm]above left:$$}](A) at (0.5,5.5) {};
\node[shape=circle,draw=black,inner sep=2pt,fill=orange!30, label={[label distance=0.1cm]above left:$$}](A) at (0.5,6) {};
\node[shape=circle,draw=black,inner sep=2pt,fill=yellow, label={[label distance=0.1cm]above left:$$}](A) at (0.5,6.5) {};

%\node[shape=circle,draw=black,inner sep=2pt,fill=black, label={[label distance=0.1cm]above left:$p_0$}](A) at (-0.5,1) {};
%\node[shape=circle,draw=black,inner sep=2pt,fill=red, label={[label distance=0.1cm]above left:$p_1$}](A) at (-0.5,0) {};
%\node[shape=circle,draw=black,inner sep=2pt,fill=red!50, label={[label distance=0.1cm]above left:$p_2$}](A) at (-0.5,2) {};
%\node[shape=circle,draw=black,inner sep=2pt,fill=blue, label={[label distance=0.1cm]above left:$p_3$}](A) at (1,2) {};
%\node[shape=circle,draw=black,inner sep=2pt,fill=blue!50, label={[label distance=0.1cm]above left:$p_4$}](A) at (-2,-1.5) {};

\end{tikzpicture}

%% file: complexity.tex
In this section, we discuss the complexity of the satisfiability problem for $\ABBL$ interpreted
over strongly discrete interval temporal structures. An EXPSPACE lower bound on the complexity follows from the reduction of the \emph{exponential-corridor tiling problem} (which is known to be EXPSPACE-complete~\cite{convenience_of_tilings}) to the satisfiability problem for the fragment $\ABB$ given in~\cite{abb_report}.

%\begin{lemma}\label{lemma:hardness}
%There is a polynomial-time reduction from the exponential-corridor tiling problem
%to the satisfiability problem for $\ABBL$.
%\end{lemma}

To give an upper bound to the complexity we claim that the existence of a compass structure (or compass generator)
$\cG$ that features a given formula $\varphi$ can be decided by verifying suitable
local (and stronger) consistency conditions over all pairs of contiguous rows, in a way similar to the EXPSPACE algorithm given in~\cite{abb_report} for \tABB. In this way, to check those local conditions it is sufficient to store only
\begin{inparaenum}[\it (i)]
	\item a counter $y$ with the number of the current row,
	\item two guessed shadings $S$ and $S'$ associated with the rows $y$ and $y+1$, and
	\item the characteristic functions of the shadings of $y$ and $y+1$.
\end{inparaenum}
Since all this information needs only an exponential amount of space, the complexity of the satisfiability problem for $\ABBL$ is in EXPSPACE.
The procedure for the infinite case is depicted in Figure \ref{fig:code}. 
For the sake of brevity, given a shading $S$ we denote with $F_S^\pi$ the unique element of $S$ such that $\req_{B}(F_S^\pi)=\emptyset$. Note that for every row $y$ with shading $S$, the type of the unit interval $[y-1,y]$ is exactly $F_S^\pi$, while the type $F$ of all other intervals in the row must contain the formula $\DB\top$, and thus it cannot be the case that $\req_{B}(F)=\emptyset$. Given a function $c_S: S \rightarrow \{0,...,8|\varphi| + 14\}$ such that $c_S(F_S^\pi)\leq 1$, we denote with $\bar{S}$ ({\em extended shading}) the pair $\langle S, c_S \rangle$; thus, in the code we use $S$ to denote a shading, and $\bar{S}$ to denote an extended-shading. Moreover we have to introduce the following stronger version of the relation $\dep{B}$:
$$
\begin{array}{rcl}
  F \localdep{B} G
  &\;\quad\text{iff}\quad\quad&
  \begin{cases}
    \req_B(F)                                              \;=\;          \obs(G) \,\cup\, \req_B(G)            \s1 \\
    \req_{\bar{B}}(G)                                      \;=\;          \obs(F) \,\cup\, \req_{\bar{B}}(F)
    \s1 \\
    \req_{\bar{L}}(F) = \req_{\bar{L}}(G).
  \end{cases}
\end{array}
$$

Finally, given two extended shadings $\bar{S}=\langle S, c_S \rangle$ and $\bar{S'}=\langle S', c_{S'} \rangle$,
we say that $\bar{S'}$  is a \emph{successor} of $\bar{S}$, and we write $\bar{S}\localdep{}\bar{S'}$, if the following conditions hold:
\begin{itemize}
\item for every   $F\in S'$ with $\req_B(F)\neq \emptyset$ there exists $G \in S$ with $F\localdep{B}G$;
\item there exists a set $R \subseteq S' \times S \times \{1,...,8|\varphi| + 14\}$ such that for every $(F,G,n)\in R$, $F\localdep{B}G$, for every $F\in S'$ we have $\sum\limits_{(F,G,n)\in R} n=c_{S'}(F)$, and for every
$G\in S$ we have $\sum\limits_{(F,G,n)\in R} n=c_{S}(G)$.
\end{itemize}
The second condition ensures that all the witnesses of the lower shading $S$ are correctly transferred in the upper shading $S'$ according to the functions $c_{S}$ and $c_{S'}$. It is easy to see that, given two rows $y$ and $y+1$ with shadings $S$ and $S'$, the two extended shadings $\bar{S} = \langle S, c_y \rangle$ and $\bar{S'} = \langle S', c_{y+1} \rangle$, (where $c_y$ and $c_{y+1}$ are the characteristic functions of $y$ and $y+1$, respectively) are such that $\bar{S}\localdep{}\bar{S'}$.

 The main procedure basically guesses two extended shadings $\bar{S}_{past}$ and $\bar{S}_{future}$ which represent
 the rows $y_0$ and $y_2$ of a compass generator, and then it checks whether a compass generator featuring them exists.  The procedure $checkPast$ ensures that we can construct  the portion of the compass structure between $y_0$ and $y_1$ (see Figure \ref{fig:compassgenerator}). The procedure starts from $y_0$ and construct this portion incrementally row by row until it reaches row $y_1$. The procedure exits successfully when it reaches, without exceeding the given number of steps, a row labelled with the extended shading $\bar{S}_{past}$ and such that all formulas $\psi \in \req_{\bar{L}}(F^\pi_{S_{past}})$ are ''witnessed'' by points with the first coordinate greater than the starting row (i.e.,  points belonging to the red triangle in Figure \ref{fig:compassgenerator}) to guarantee that there exists a past witness set for $y_1$ that respects condition \textbf{G3} of Definition~\ref{def:compassgenerator}. This condition is verified by means of the set $S_{lower}$ which keeps track of such points. The procedure $checkFinite$ simply checks if the extended shading $\bar{S}_{future}$ is ''reachable'' from the extended shading $\bar{S}_{past}$, and thus it represents the construction of the finite part of a compass generator (the portion between $y_1$ and $y_2$ in Figure \ref{fig:compassgenerator}). Finally the the procedure $checkFuture$ ensures that we can construct  the portion between $y_2$ and $N-1$ of a compass generator.
 This last procedure is similar to the procedure $checkPast$, and it checks whether there exists a portion of a compass structure where both the lowest and the biggest rows are labelled with $\bar{S}_{future}$. To guarantee that a future witness set for $y_2$ exists (condition \textbf{G6} of Definition~\ref{def:compassgenerator}), we require that for every $F\in S_{future}$ and for every $\psi \in \req_{\bar{B}}(F)$, it is the case that $\psi$ is fulfilled by some successor of $S_{future}$. This condition is ensured by means of the set $REQ_F$, which keeps track of the formulas in $\req_{\bar{B}}(F)$ that still need to be satisfied. It is worth to notice that all the counters, the extended shadings, and the shadings using in these procedures can be represented using exponential space with respect to the length of the input formula.
%
%
%\begin{lemma}\label{lemma:completeness}
%There is an EXPSPACE non-deterministic procedure that decides whether a given formula
%of $\ABBL$ is satisfiable or not.
%\end{lemma}
%
%\noindent
Summing up, we obtain the following tight complexity result.

\begin{theorem}\label{th:complexity}
The satisfiability problem for $\ABBL$ interpreted over strongly complete linear orders is EXPSPACE-complete.
\end{theorem}

\begin{figure}[t!]
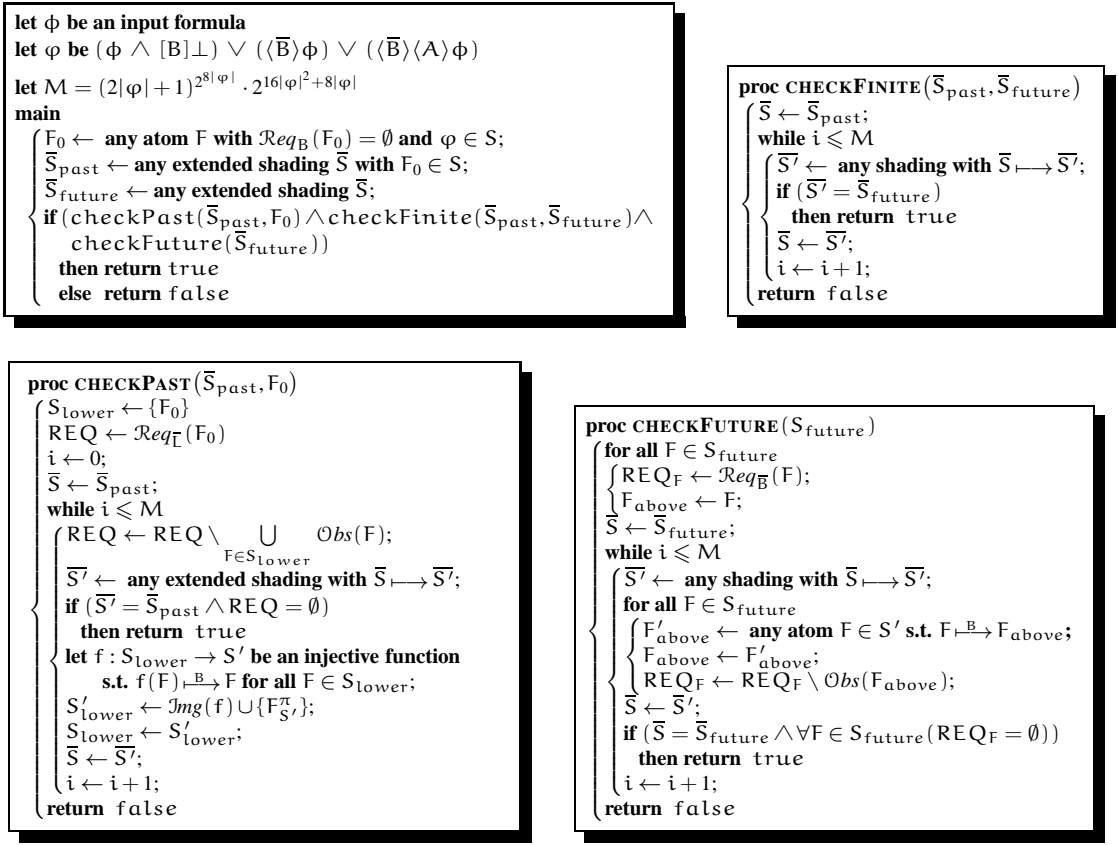

\scriptsize
\centering
\begin{code}{87mm}

 \LET \phi \text{ be an input formula} \\
 \LET \varphi \text{ be } (\phi\et\BB\bot)\vel(\DBbar\phi)\vel(\DBbar\DA\phi)   \s1 \\
 \LET M =(2|\varphi|+1)^{2^{8\len{\varphi}}}\cdot 2^{16\len{\varphi}^2 + 8\len{\varphi}}\\

\MAIN
\BEGIN
F_0 \gets \text{ any atom $F$ with $\req_{B}(F_0)=\emptyset$ and $\varphi \in S$};\\
\bar{S}_{past} \gets \text{any extended shading $\bar{S}$ with $F_0 \in S$};\\
\bar{S}_{future} \gets \text{any extended shading $\bar{S}$};\\
\IF\!(checkPast(\bar{S}_{past}, F_0)\wedge checkFinite(\bar{S}_{past},\bar{S}_{future})\wedge \\  \ \ \ \ \ checkFuture(\bar{S}_{future}) )
\THEN \RETURN\  true \ELSE \RETURN\  false
\END
\ENDMAIN
\end{code}
\quad
\begin{code}{48mm}

\FUNCTION{checkFinite}{\bar{S}_{past},\bar{S}_{future}}
\BEGIN
\bar{S}\gets \bar{S}_{past};\\
\WHILE i\leq M\\
\BEGIN
\bar{S'}\gets \text{ any shading with $\bar{S}\localdep{} \bar{S'}$};\\
\IF (\bar{S'}=\bar{S}_{future})
\THEN \RETURN \ \ true\\
\bar{S} \gets \bar{S'};\\
i\gets i+1;\\
\END\\
\RETURN \ \ {false}
\END

\ENDFUNCTION
\end{code}

\bigskip

\begin{code}{67mm}
\FUNCTION{checkPast}{\bar{S}_{past}, F_0}
\BEGIN
S_{lower} \gets \{F_0\} \\
REQ \gets \req_{\bar{L}}(F_0)\\
i\gets 0;\\
\bar{S}\gets \bar{S}_{past};\\
\WHILE i\leq M\\
\BEGIN
REQ \gets REQ \setminus \bigcup\limits_{F\in S_{lower}} \obs(F);\\
\bar{S'}\gets \text{ any extended shading with $\bar{S}\localdep{} \bar{S'}$};\\
\IF (\bar{S'}=\bar{S}_{past}\wedge REQ =\emptyset)
\THEN \RETURN \ \ true\\
\LET f:S_{lower} \rightarrow S' \text{ be an injective function}\\ \text{\ \ \ \ \ \ \ \medskip s.t. $f(F)\localdep{B}F$ for all $F\in S_{lower}$};\\
S'_{lower} \gets \img(f) \cup\{F^\pi_{S'}\};\\
S_{lower} \gets S'_{lower};\\
\bar{S} \gets \bar{S'};\\
i\gets i+1;\\
\END\\
\RETURN \ \ {false}
\END

\ENDFUNCTION
\end{code}
\quad
\begin{code}{67mm}
\FUNCTION{checkFuture}{S_{future}}
\BEGIN
\FORALL F\in S_{future}\\
\BEGIN
 REQ_F \gets \req_{\bar{B}}(F) ;\\
F_{above} \gets F  ;\\
\END\\
\bar{S}\gets \bar{S}_{future};\\
\WHILE i\leq M\\
\BEGIN
\bar{S'}\gets \text{ any shading with $\bar{S}\localdep{} \bar{S'}$};\\
\FORALL F \in S_{future}\\
\BEGIN
F'_{above}\gets \text{ any atom $F\in S'$ s.t. $F\localdep{B} F_{above}$; }\\
F_{above} \gets F'_{above};\\
REQ_F\gets REQ_F \setminus \obs(F_{above});\\
\END\\
\bar{S}\gets \bar{S}';\\
\IF (\bar{S}=\bar{S}_{future}\wedge \forall F\in S_{future}  (REQ_F=\emptyset))
\THEN \RETURN \ \ true\\
i\gets i+1;\\
\END\\
\RETURN \ \ {false}
\END

\ENDFUNCTION
\end{code}
\label{fig:code}
\caption{the procedure for checking the satisfiability of $\phi$ over the integers.}
\end{figure}

%% file: conclusions.tex
We considered an interval temporal logic (\tABBL) with four modalities, corresponding, respectively, to Allen's interval relations \emph{after}, \emph{begins}, \emph{begun-by}, and \emph{before}, and interpreted in the class of all strongly discrete linearly ordered sets, which includes, among others, all frames built over $\bbN$, $\bbZ$, and finite orders. We showed that this logic is decidable in EXPSPACE, and complete for this class. The importance of this result relies on the fact that, for the considered interpretations, this logic is maximal with respect to decidability. Moreover, these results represent a non-trivial contribution towards the complete classification of all fragments of Halpern and Shoham's modal logic of intervals. We plan to complete the study of this particular language when it is interpreted over other classes of orders, such as the class of all dense linearly ordered sets, or the class of all linear orders, and to refine these results to include point-intervals, too.

%% file: draftGANDALF2010.bbl
\begin{thebibliography}{10}
\providecommand{\bibitemstart}[1]{\bibitem{#1}}
\providecommand{\bibitemend}{}
\providecommand{\bibliographystart}{}
\providecommand{\bibliographyend}{}
\providecommand{\url}[1]{\texttt{#1}}
\providecommand{\urlprefix}{Available at }
\providecommand{\bibinfo}[2]{#2}
\bibliographystart

\bibitemstart{interval_relations}
\bibinfo{author}{J.F. Allen} (\bibinfo{year}{1983}):
  \emph{\bibinfo{title}{Maintaining Knowledge About Temporal Intervals}}.
\newblock {\sl \bibinfo{journal}{Communications of the Association for
  Computing Machinery}} \bibinfo{volume}{26}(\bibinfo{number}{11}), pp.
  \bibinfo{pages}{832--843}.
\bibitemend

\bibitemstart{halpern_shoham_fragments}
\bibinfo{author}{D.~Bresolin}, \bibinfo{author}{D.~{Della Monica}},
  \bibinfo{author}{V.~Goranko}, \bibinfo{author}{A.~Montanari} \&
  \bibinfo{author}{G.~Sciavicco} (\bibinfo{year}{2008}):
  \emph{\bibinfo{title}{Decidable and Undecidable Fragments of {H}alpern and
  {S}hoham's Interval Temporal Logic: towards a complete classification}}.
\newblock In: {\sl \bibinfo{booktitle}{Proceedings of the 15th International
  Conference on Logic for Programming, Artificial Intelligence, and Reasoning
  (LPAR)}}, {\sl \bibinfo{series}{Lecture Notes in Computer Science}}
  \bibinfo{volume}{5330}, \bibinfo{publisher}{Springer}, pp.
  \bibinfo{pages}{590--604}.
\bibitemend

\bibitemstart{pnl_expressiveness}
\bibinfo{author}{D.~Bresolin}, \bibinfo{author}{V.~Goranko},
  \bibinfo{author}{A.~Montanari} \& \bibinfo{author}{G.~Sciavicco}
  (\bibinfo{year}{2009}): \emph{\bibinfo{title}{Propositional Interval
  Neighborhood Logics: expressiveness, decidability, and undecidable
  extensions}}.
\newblock {\sl \bibinfo{journal}{Annals of Pure and Applied Logic}}
  \bibinfo{volume}{161}(\bibinfo{number}{3}), pp. \bibinfo{pages}{289--304}.
\bibitemend

\bibitemstart{pnl_logics}
\bibinfo{author}{V.~Goranko}, \bibinfo{author}{A.~Montanari} \&
  \bibinfo{author}{G.~Sciavicco} (\bibinfo{year}{2003}):
  \emph{\bibinfo{title}{Propositional Interval Neighborhood Temporal Logics}}.
\newblock {\sl \bibinfo{journal}{Journal of Universal Computer Science}}
  \bibinfo{volume}{9}(\bibinfo{number}{9}), pp. \bibinfo{pages}{1137--1167}.
\bibitemend

\bibitemstart{roadmap_intervals}
\bibinfo{author}{V.~Goranko}, \bibinfo{author}{A.~Montanari} \&
  \bibinfo{author}{G.~Sciavicco} (\bibinfo{year}{2004}):
  \emph{\bibinfo{title}{A Road Map of Interval Temporal Logics and Duration
  Calculi}}.
\newblock {\sl \bibinfo{journal}{Applied Non-classical Logics}}
  \bibinfo{volume}{14}(\bibinfo{number}{1-2}), pp. \bibinfo{pages}{9--54}.
\bibitemend

\bibitemstart{interval_modal_logic}
\bibinfo{author}{J.Y. Halpern} \& \bibinfo{author}{Y.~Shoham}
  (\bibinfo{year}{1991}): \emph{\bibinfo{title}{A Propositional Modal Logic of
  Time Intervals}}.
\newblock {\sl \bibinfo{journal}{Journal of the ACM}} \bibinfo{volume}{38}, pp.
  \bibinfo{pages}{279--292}.
\bibitemend

\bibitemstart{abba_finite}
\bibinfo{author}{A.~Montanari}, \bibinfo{author}{G.~Puppis} \&
  \bibinfo{author}{P.~Sala} (\bibinfo{year}{2010}):
  \emph{\bibinfo{title}{Maximal decidable fragments of Halpern and Shoham's
  modal logic of intervals}}.
\newblock In: {\sl \bibinfo{booktitle}{Proceedings of the 37th International
  Colloquium on Automata, Languages and Programming (ICALP 2010). To appear.}}
\bibitemend

\bibitemstart{abb_report}
\bibinfo{author}{A.~Montanari}, \bibinfo{author}{G.~Puppis},
  \bibinfo{author}{P.~Sala} \& \bibinfo{author}{G.~Sciavicco}
  (\bibinfo{year}{2009}): \emph{\bibinfo{title}{Decidability of the Interval
  Temporal Logic {$AB\bar{B}$} Over the Natural Numbers}}.
\newblock \bibinfo{type}{Research Report} \bibinfo{number}{{UDMI}/2009/07},
  \bibinfo{institution}{Department of Mathematics and Computer Science,
  University of Udine}, \bibinfo{address}{Udine, Italy}.
\newblock
  \bibinfo{note}{\texttt{http://users.dimi.uniud.it/$\sim$angelo.montanari/rr2%
00907.pdf}}.
\bibitemend

\bibitemstart{abb_natural}
\bibinfo{author}{A.~Montanari}, \bibinfo{author}{G.~Puppis},
  \bibinfo{author}{P.~Sala} \& \bibinfo{author}{G.~Sciavicco}
  (\bibinfo{year}{2010}): \emph{\bibinfo{title}{Decidability of the Interval
  Temporal Logic {$AB\bar{B}$} on Natural Numbers}}.
\newblock In: {\sl \bibinfo{booktitle}{Proceedings of the 27th Symposium on
  Theoretical Aspects of Computer Science (STACS 2010)}}, pp.
  \bibinfo{pages}{597--608}.
\bibitemend

\bibitemstart{convenience_of_tilings}
\bibinfo{author}{P.~{Van Emde Boas}} (\bibinfo{year}{1997}):
  \emph{\bibinfo{title}{The Convenience of Tilings}}.
\newblock In: {\sl \bibinfo{booktitle}{Complexity, Logic and Recursion
  Theory}}, {\sl \bibinfo{series}{Lecture Notes in Pure and Applied
  Mathematics}} \bibinfo{volume}{187}, \bibinfo{publisher}{Marcel Dekker Inc.},
  pp. \bibinfo{pages}{331--363}.
\bibitemend

\bibliographyend
\end{thebibliography}
